\documentclass[a4paper,UKenglish,cleveref, autoref, thm-restate,authorcolumns]{lipics-v2019}
\pdfoutput=1

\title{Variable Shift SDD: A More Succinct Sentential Decision Diagram} 

\titlerunning{Variable Shift SDD} 

\author{Kengo Nakamura}{NTT Communication Science Laboratories, Kyoto, Japan}{kengo.nakamura.dx@hco.ntt.co.jp}{}{}
\author{Shuhei Denzumi}{Graduate School of Information Science and Technology, The University of Tokyo, Tokyo, Japan}{denzumi@mist.i.u-tokyo.ac.jp}{}{}
\author{Masaaki Nishino}{NTT Communication Science Laboratories, Kyoto, Japan}{masaaki.nishino.uh@hco.ntt.co.jp}{}{}

\authorrunning{K. Nakamura, S. Denzumi, and M. Nishino} 

\Copyright{Kengo Nakamura, Shuhei Denzumi, and Masaaki Nishino} 

\ccsdesc[100]{Theory of computation~Data structures design and analysis}
\ccsdesc[100]{Computing methodologies~Knowledge representation and reasoning}

\keywords{Boolean function, Data structure, Sentential decision diagram} 

\category{} 

\relatedversion{} 

\supplement{}

\acknowledgements{}

\hideLIPIcs  

\nolinenumbers

\EventEditors{John Q. Open and Joan R. Access}
\EventNoEds{2}
\EventLongTitle{42nd Conference on Very Important Topics (CVIT 2016)}
\EventShortTitle{CVIT 2016}
\EventAcronym{CVIT}
\EventYear{2016}
\EventDate{December 24--27, 2016}
\EventLocation{Little Whinging, United Kingdom}
\EventLogo{}
\SeriesVolume{42}
\ArticleNo{23}

\DeclareMathOperator*{\argmin}{arg\,min}

\usepackage{pifont}
\usepackage{textcomp}
\usepackage{algorithm}
\usepackage[noend]{algorithmic}

\newcommand{\mathtrue}{\mathrm{\textit{true}}}
\newcommand{\mathfalse}{\mathrm{\textit{false}}}
\newcommand{\apos}{\mathbf{v}}
\newcommand{\aneg}{\neg\mathbf{v}}
\newcommand{\mathapply}{\mathtt{Apply}}

\newcommand{\vtoid}{\mathtt{ID}}
\newcommand{\idtov}{\mathtt{ID}^{-1}}

\newcommand{\order}[1]{O (#1)}
\newcommand{\morder}[1]{\Omega (#1)}

\newcommand{\myarrow}{\mathord{\uparrow}}
\renewcommand{\epsilon}{\varepsilon}

\begin{document}

\maketitle

\begin{abstract}
	The Sentential Decision Diagram (SDD) is a tractable representation of Boolean functions that subsumes the famous Ordered Binary Decision Diagram (OBDD) as a strict subset.
	SDDs are attracting much attention because they are more succinct than OBDDs, as well as having canonical forms and supporting many useful queries and transformations such as model counting and $\texttt{Apply}$ operation.
	In this paper, we propose a more succinct variant of SDD named \emph{Variable Shift SDD} (VS-SDD).
	The key idea is to create a unique representation for Boolean functions that are equivalent under a specific variable substitution.
	We show that VS-SDDs are never larger than SDDs and there are cases in which the size of a VS-SDD is exponentially smaller than that of an SDD. 
	Moreover, despite such succinctness, we show that numerous basic operations that are supported in polytime with SDD are also supported in polytime with VS-SDD.
	Experiments confirm that VS-SDDs are significantly more succinct than SDDs when applied to classical planning instances, where inherent symmetry exists. 
\end{abstract}

\section{Introduction}
\label{sec:intro}

The succinct representations of a Boolean function have long been studied in the computer science community.
Among them, the \emph{Ordered Binary Decision Diagram} (OBDD)~\cite{bryant86bdd} has been used as a prominent tool in various applications.
An OBDD represents a Boolean function as a directed acyclic graph (DAG).
The reason for the popularity of OBDDs is that it can often represent a Boolean function very succinctly while supporting many useful queries and transformations in polytime with respect to the compilation size.

In the last few years, the \emph{Sentential Decision Diagram} (SDD)~\cite{darwiche11sdd}, which is also a DAG representation, has also attracted attention~\cite{vlasse14LP,oztok15topdown}.
SDDs have a tighter bound on the compilation size than OBDDs~\cite{darwiche11sdd}, and there are cases in which the use of SDDs can make the size exponentially smaller than OBDDs~\cite{bova16exp}.
In addition, SDDs also support a number of queries and transformations in polytime.
Among them, the most important polytime operation is the Apply operation, which takes two SDDs representing two Boolean functions $f, g$ and binary operator $\circ$, such as conjunction ($\wedge$) and disjunction ($\vee$), and returns the SDD representing the Boolean function $f\circ g$.
This operation is fundamental in computing an arbitrary Boolean function into an SDD, as well as in proving the polytime solvability of various important and useful operations.

One of the reasons why OBDDs and SDDs, as well as many other such DAG representations,
can express a Boolean function succinctly is that they share identical substructures that represent the equivalent
Boolean function; they represent a Boolean function by recursively
decomposing it into subfunctions that can also be represented as DAGs. If
a decomposition generates equivalent subfunctions, we do not
need to have multiple DAGs, and thus more succinctly represent the original Boolean
function. Since the effectiveness of such representations depend on the DAG size, 
representations that are more succinct while still supporting useful operations are always in demand.

In this paper, we propose a new SDD-based structure named \emph{Variable Shift SDD} (VS-SDD); it can even more succinctly represent Boolean functions, while supporting polytime $\mathapply$ operations.
The key idea is to extend the condition for sharing
DAGs. While an SDD can share DAGs representing identical Boolean functions, a VS-SDD can  share DAGs representing
Boolean functions that are equivalent under a specific variable
substitution.
For example, consider two Boolean functions $f = A \wedge B$ and $g = C \wedge
D$ defined over  variables $A, B, C, D$. An SDD cannot share
DAGs representing $f$ and $g$ since they are  not equivalent. On the other hand, VS-SDD can share them since
$f$ and $g$ are equivalent under the variable substitution that exchanges $A$ with $C$ and
$B$ with $D$. Such Boolean functions  appear
in a wide range of situations. One typical example is modeling
time-evolving systems; such as, we want to find a sequence of assignments of variables
$\mathbf{x}_1, \ldots, \mathbf{x}_T$ over timestamps $t = 1, \ldots, T$
such that every $\mathbf{x}_t$ satisfies the condition that $h(\mathbf{x}_t) =
\mathtrue$. Such a sequence is modeled as Boolean function $f(\mathbf{X}_1, \ldots, \mathbf{X}_T) =
h^{(1)}(\mathbf{X}_1) \wedge \cdots \wedge h^{(T)}(\mathbf{X}_T)$, where $h^{(t)}$
is $h(\mathbf{X})$ defined over $\mathbf{X}_t$.
Since all $h^{(t)}(\mathbf{X}_t)$ are equivalent under variable substitutions, it is highly
possible that VS-SDD can yield more succinct representations.

Technically, these advantages of VS-SDD are obtained by introducing the indirect
specification of depending variables. Every SDD is associated with a set of variables that the corresponding Boolean function depends on.
In SDD, such set of variables are represented by IDs, where each set of variables has a unique ID.
On the other hand, VS-SDD represents such sets of variables by storing the \emph{difference} of IDs.
This allows the sharing of the Boolean functions that are equivalent under specific types of variable substitutions. 

Our main results are as follows:
\begin{itemize}
	\item VS-SDDs are never larger their SDDs equivalents. Moreover, there is a
              class of Boolean functions for which VS-SDDs are exponentially smaller than SDDs.
	\item  VS-SDD supports polytime $\mathapply$. Moreover, the queries
               and transformations listed in \cite{darwiche02KC} that SDDs
               support in polytime are also supported in polytime by
               VS-SDDs.
        \item  We experimentally confirm that when applied to classical planning instances, VS-SDDs are significantly smaller than SDDs.
\end{itemize}
To summarize, VS-SDDs incur no additional overhead over SDDs while being potentially much smaller than SDDs.

The rest of this paper is organized as follows.
Sect.~\ref{sec:related} reviews related works.
Sect.~\ref{sec:preliminary} gives the preliminaries.
Sect.~\ref{sec:SDD} introduces SDD, on which our proposed structure is based.
Sect.~\ref{sec:VSSDD} describes the formal definition of the equivalence relation we want to share, the definition of VS-SDD, and the relation between them.
Sect.~\ref{sec:property} examines the properties of VS-SDDs.
Sect.~\ref{sec:operation} deals with the operations on VS-SDDs, especially $\mathapply$.
Sect.~\ref{sec:implement} mentions some implementation details that ensure that VS-SDDs suffer no overhead penalty relative to SDDs.
Sect.~\ref{sec:eval} provides experiments and their results, and
Sect.~\ref{sec:conclusion} gives concluding remarks.

\section{Related Works}
\label{sec:related}
There have been studies that attempt to share the substructures that represent the ``equivalent'' Boolean functions up to a conversion.
For OBDDs, the most famous among them are complement edges and attributed edges~\cite{madre88ce,minato90vs}.
For example, with complement edges, we can share the substructures representing the equivalent Boolean functions up to taking a negation.
However, this study does not focus on the solvability of the operations in the compressed form. Actually, some $\mathapply$ operations cannot be performed 
in a compressed form.
After that, the differential BDD~\cite{anuchi95diff}, especially $\myarrow\Delta$BDD, was proposed to share equivalent Boolean functions up to the shifting of variables, that is, given the total order of the variables, shift them uniformly to share isomorphic substructures.
This structure supports operations like $\mathapply$, but its complexity depends on the number of variables, which means that this operation is not supported in polytime of the compilation size.
With regard to other representations, Sym-DDG/FBDD~\cite{bart14sym}, based on DDG~\cite{fragier06DDG} and FBDD~\cite{gergov94FBDD}, can share equivalent functions up to variable substitution.
Since their method adopts a permutation of variables, it can, in principle, treat any variable substitution.
However, Sym-DDG/FBDD fails to support some important operations such as conditioning and $\mathapply$.
With regard to these previous works, VS-SDD differs in three points.
First, to the best of our knowledge, VS-SDD is the first attempt to extend the equivalence relationships of an SDD.
We should note that VS-SDD is not obtained by a straightforward application of the techniques invented for OBDDs.
Second, VS-SDD has theoretical guarantees on its size.
Last, it supports the flexible polytime $\mathapply$ operation.

\section{Preliminaries}
\label{sec:preliminary}
We use an uppercase letter (e.g., $X$) to represent a variable and a lowercase letter (e.g., $x$) to denote its assignment (either $\mathtrue$ or $\mathfalse$).
A bold uppercase letter (e.g., $\mathbf{X}$) represents a set of variables and a bold lowercase letter (e.g., $\mathbf{x}$) denotes its assignment. 
\emph{Boolean function} $f(\mathbf{X})$ is a function that maps each assignment of $\mathbf{X}$ to either $\mathtrue$ or $\mathfalse$. The \emph{conditioning} of $f$ on instantiation $\mathbf{X}$, written $f | \mathbf{x}$, is the subfunction that results from setting variables $\mathbf{X}$ to their values in $\mathbf{x}$. We say $f$ \emph{essentially depends} on variable $X$ iff $f | X \neq f | \neg X$.
We take $f(\mathbf{Z})$ to mean that $f$ can only essentially depend on variables in $\mathbf{Z}$.
A \emph{trivial} function maps all its inputs to 0 (denoted \emph{false}) or maps all to 1 (denoted \emph{true}).

Consider an ordered full binary tree.
For two nodes $u,w$ in it, we say $w$ is a \emph{left descendant} (resp.~\emph{right descendant}) of $u$ if $w$ is a (not necessarily proper) descendant of the left (resp.~right) child of $u$.

\section{Sentential Decision Diagrams}
\label{sec:SDD}
First, we introduce SDD.
It is a data structure that can represent a Boolean function as a directed acyclic graph (DAG) like OBDD.

Let $f$ be a Boolean function and $\mathbf{X}, \mathbf{Y}$ be non-intersecting sets of  variables.
The $(\mathbf{X}, \mathbf{Y})$\emph{-decomposition} of $f$ is 
$f=\bigvee_{i=1}^{n}[p_i(\mathbf{X})\wedge s_i(\mathbf{Y})]$,
where $p_i(\mathbf{X})$ and $s_i(\mathbf{Y})$ are Boolean functions.
Here $p_1,\ldots,p_n$ are called \emph{primes} and $s_1,\ldots,s_n$ are called \emph{subs}.
We denote $(\mathbf{X}, \mathbf{Y})$-decomposition as $\{(p_1, s_1), \ldots, (p_n, s_n)\}$, where $n$ is the size of the decomposition. The pair $(p_i, s_i)$ is called an \emph{element}.
An $(\mathbf{X},\mathbf{Y})$-decomposition is called $\mathbf{X}$\emph{-partition} iff $p_i\wedge p_j=\mathfalse$ for all $i\neq j$, $\bigvee_{i=1}^{n}p_i=\mathtrue$, and $p_i\neq\mathfalse$ for all $i$.
If $s_i\neq s_j$ for all $i\neq j$, the partition is called \emph{compressed}.
It is known that a function $f(\mathbf{X}, \mathbf{Y})$ has exactly one compressed $\mathbf{X}$-partition (see Theorem 3 of \cite{darwiche11sdd}).

An SDD decomposes a Boolean function by recursively applying $\mathbf{X}$-partitions.
The structure of partitions is determined by an ordered full binary tree called the \emph{vtree}; its leaves have a one-to-one correspondence with variables.
Here, each internal node partitions the variables into those in the left subtree ($\mathbf{X}$) and those in the right subtree ($\mathbf{Y}$).
For example, the vtree in Fig.~\ref{fig:sdds}(a) shows the recursive partition of
variables $A, B, C, D$. The root node represents the partition of variables to
$\{A, B\}, \{C, D\}$, while the left child of the root represents the partition
$\{A\}, \{B\}$. SDD implements $\mathbf{X}$-partitions by following these recursive partitions of variables.

\begin{figure}[tbp]
	\centering
	\includegraphics[keepaspectratio,scale=0.75]{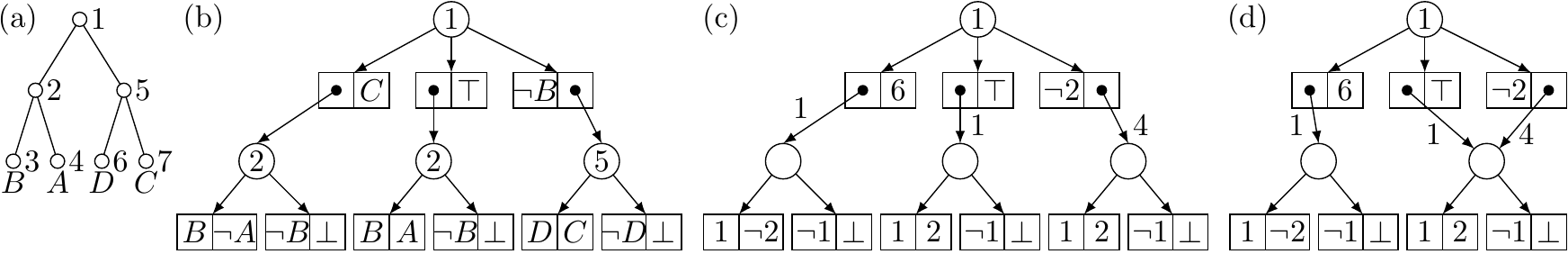}
	\caption{(a) An example of a vtree. (b) The SDD of $f=(A\wedge B)\vee(B\wedge C)\vee(C\wedge D)$ that respects the vtree of (a). (c)(d) The VS-SDD of $f=(A\wedge B)\vee(B\wedge C)\vee(C\wedge D)$ given the vtree (a) with offset $1$. Here (d) is the more reduced form than (c).}
	\label{fig:sdds}
\end{figure}

Let $\langle\cdot\rangle$ be a mapping from an SDD to a Boolean function (i.e., the semantics of SDD).
The SDD is defined recursively as follows.
\begin{definition}
	\label{def:SDD}
	The following $\alpha$ is an SDD that respects vtree node $v$.
	\begin{itemize}
		\item \emph{(constant)} $\alpha=\top$ or $\alpha=\bot$.
		Semantics: $\langle\top\rangle =\mathtrue$ and $\langle\bot\rangle =\mathfalse$.
		\item \emph{(literal)} $\alpha=X$ or $\alpha=\neg X$, and $v$ is a leaf  node with variable $X$.
		Semantics: $\langle X\rangle = X$ and $\langle\neg X\rangle = \neg X$.
		\item \emph{(decomposition)} $\alpha=\{(p_1,s_1),\ldots,(p_n,s_n)\}$, and $v$ is an internal node.
		Here each $p_i$ is an SDD respecting a left descendant node of $v$, each $s_i$ is an SDD respecting a right descendant node of $v$, and $\langle p_1\rangle,\ldots,\langle p_n\rangle$ form a partition.
		Semantics: $\langle\alpha\rangle = \bigvee_{i=1}^{n}[\langle p_i\rangle\wedge\langle s_i\rangle]$.
	\end{itemize}
	The size of $\alpha$ (denoted by $|\alpha|$) is defined as the sum of the sizes of all its decompositions.
\end{definition}

Given the vtree of Fig.~\ref{fig:sdds}(a), Fig.~\ref{fig:sdds}(b) depicts an SDD that respects the vtree node labeled 1 and represents $f=(A\wedge B)\vee(B\wedge C)\vee(C\wedge D)$.
At the top level, $f$ is decomposed as $[(\neg A\wedge B)\wedge C]\vee [(A\wedge B)]\vee [\neg B\wedge(C\wedge D)]$. This is the compressed $\{A, B\}$-partition since primes $\neg A\wedge B$, $A \wedge B$ and $\neg B$ satisfy the condition for $\{A, B\}$-partition and subs are all different.
Here each circle represents a decomposition node, and the number inside each circle indicates the respecting vtree node ID. The size of the SDD is  9.

There are two classes of canonical SDDs.
We say a class of SDDs is canonical iff, given a vtree, for any Boolean function $f$, there is exactly one SDD in this class that represents $f$.
Here we consider only \emph{reduced} SDDs, i.e.~the SDDs such that the identical substructures are fully merged.
\begin{definition}
	\label{def:SDDcharacter}
	We say SDD $\alpha$ is \emph{compressed} iff all partitions in $\alpha$ are compressed.
	We say $\alpha$ is \emph{trimmed} iff it does not have decompositions of the form $\{(\top,\beta)\}$ and $\{(\beta,\top),(\neg\beta,\bot)\}$, and \emph{lightly trimmed} iff it does not have decompositions of the form $\{(\top,\top)\}$ and $\{(\top,\bot)\}$.
	We say $\alpha$ is \emph{normalized} iff for each decomposition that respects vtree node $w$, its primes respect the left child of $w$ and its subs respect the right child of $w$.
\end{definition}
\begin{theorem}[\cite{darwiche11sdd}]
	\label{thm:SDDcanonical}
	Compressed and trimmed SDDs are canonical.
	Also, compressed, lightly trimmed, and normalized SDDs are canonical.
\end{theorem}

The key property of SDDs is that they support the polytime $\mathapply$ operation, which takes, given vtree $v$, two SDDs $\alpha, \beta$ and binary operation $\circ$, and computes a new SDD that represents $\langle\alpha\rangle\circ\langle\beta\rangle$ in $\order{|\alpha||\beta|}$ time.
Using $\mathapply$, we can compile an arbitrary Boolean function into an SDD.

\section{Variable Shift Sentential Decision Diagrams}
\label{sec:VSSDD}
We now introduce our more succinct variant of SDD, named \emph{variable shift SDD} (VS-SDD).
As mentioned above, an SDD expresses a Boolean function succinctly by sharing equivalent substructures that represent the same Boolean subfunction.
The motivation to introduce VS-SDD is, in addition to this, to share substructures that represent equivalent Boolean functions under a particular variable substitution.

First, we briefly describe the idea by using an intuitive example.
Let us consider two Boolean functions $f  = X_1 \wedge X_2$ and $g  = X_3 \wedge X_4$ defined over variables $X_1, \ldots, X_4$. 
Apparently, $f$ and $g$ are not equivalent, but they are equivalent  if we exchange $X_1$ with $X_3$ and $X_2$ with $X_4$. 
We formally define this equivalency of Boolean functions below.
\begin{definition}
    \label{def:subst-equiv}
    We say two Boolean functions $f$, $g$ defined over $\mathbf{X}$ are     \emph{substitution-equivalent} with permutation $\pi$ if 
    $f(X_1 = x_1, \ldots, X_M = x_M) = g(X_1 = x_{\pi(1)}, \ldots, X_M = x_{\pi(M)})$ for any assignment $\mathbf{x}$, where $M = |\mathbf{X}|$  and $\pi : \{1, \ldots, M\} \mapsto \{1, \ldots, M\}$ is a bijection.
\end{definition}

In the above example, $f$ and $g$ are substitution equivalent with $\pi$ satisfying $\pi(3) = 1$ and $\pi(4) = 2$. 
For $i=1,2$, this permutation is defined by simply adding constant to an input, i.e.,  $\pi(i) = i+c\ (i=1,2)$ where the constant $c = 2$.
This result implies that a class
of substitution-equivalent functions can be represented
as the pair of a base representation and constant value $c$.
VS-SDD exploits this idea.

\subsection{Definition of the structure}
Now we consider the structure and semantics of VS-SDD.
VS-SDD shares many properties with SDDs; it is defined with  a vtree and a DAG structure representing recursive $\mathbf{X}$-partitions following
the vtree. VS-SDD has two main differences from SDD.
First, it associates every vtree node with an integer ID and it considers some mathematical operations over them. 
We use $\vtoid(v)$ to represent the ID associated with vtree node $v$, and $\idtov(i)$ to represent the vtree node that corresponds to ID $i$. In the following, we  assume that integer IDs of vtree nodes are assigned following a preorder traversal of the vtree.
The IDs assigned to the vtree in Fig.~\ref{fig:sdds}(a) satisfy this condition.
Second, while SDD represents a Boolean function as a node of a DAG, VS-SDD represents a Boolean function as a pair $(\alpha, k)$ of node $\alpha$ in a DAG  and integer $k$. We say $\alpha$ is the VS-SDD \emph{structure} and $k$ is its \emph{offset}. We use $\langle \alpha, k \rangle$ as a mapping from VS-SDD $(\alpha, k)$  to the corresponding Boolean function. 
\begin{definition}
	Given vtree $v$, the following $(\alpha, k)$ is a VS-SDD.
	\begin{itemize}
		\item \emph{(constant)} $\alpha=\top$ or $\alpha=\bot$.
		Semantics: $\langle\top,\cdot\rangle=\mathtrue$ and $\langle\bot,\cdot\rangle=\mathfalse$.
		\item \emph{(literal)} $\alpha=\apos$ or $\alpha=\aneg$, and $\idtov(k)$ is a leaf vtree node.
		Semantics: $\langle\apos,k\rangle=l(\idtov(k))$ and $\langle\aneg,k\rangle=\neg l(\idtov(k))$, where $l(v)$ is a variable corresponding to vtree node $v$.
		\item \emph{(decomposition)} $\alpha=\{([p_1, d_1], [s_1, e_1]),\ldots,([p_n, d_n],[s_n, e_n])\}$, and $\idtov(k)$ is an internal node of $v$.
		Here each $p_i$ is a VS-SDD structure and $d_i$ is an integer such that $\idtov(d_i+k)$ is a left descendant vtree node of $\idtov(k)$. Similarly, each $s_i$ is a VS-SDD structure and $e_i$ is integer such that $\idtov(e_i+k)$ is a right descendant node of $\idtov(k)$ and Boolean functions $\langle p_1  , d_1 + k\rangle,\ldots,\langle p_n, d_n + k \rangle$ form a partition.
		Semantics: $\langle\alpha,k\rangle = \bigvee_{i=1}^{n}(\langle p_i,d_i+k\rangle\wedge\langle s_i,e_i+k\rangle)$.
	\end{itemize}
	The size of $\alpha$ (denoted by $|\alpha|$) is defined as the sum of the sizes of all decompositions.
\end{definition}

Given the vtree of Fig.~\ref{fig:sdds}(a), Fig.~\ref{fig:sdds}(c)-(d) depict the VS-SDDs representing $f=(A\wedge B)\vee(B\wedge C)\vee(C\wedge D)$, where Fig.~\ref{fig:sdds}(d) is a further reduced form created by sharing the identical substructures in Fig.~\ref{fig:sdds}(c). Here the offset is written in the circle of the root node.
Every prime $[p_i, d_i]$ is drawn as an arrow to structure $p_i$ annotated with $d_i$, except for the following cases.
If $p_i = \apos$ (resp.~$\aneg$), it is drawn as simply $d_i$ (resp.~$\neg d_i$).
If $p_i$ is either of $\top$ or $\bot$, it is represented by $p_i$ itself, since the value of $d_i$ has no effect on the semantics.
Subs $[s_i, e_i]$ are treated in the same way.

We first give an interpretation of VS-SDD. 
By comparing the SDD in Fig.~\ref{fig:sdds}(b) with the VS-SDD in Fig.~\ref{fig:sdds}(c) having the same structure, we find they differ only in the labels
of nodes and edges. Actually, we can construct the SDD of Fig.~\ref{fig:sdds}(b) from the VS-SDD in Fig.~\ref{fig:sdds}(c) in the following way. Let $P_\alpha$ be a path from the root to VS-SDD structure $\alpha$ and
$D_{P_\alpha}$ be the sum of the offset and edge values appearing along the path. Then, $\idtov(D_{P_\alpha})$ is the vtree node that the corresponding SDD node respects. For example, the leftmost child of the root node in the VS-SDD in Fig.~\ref{fig:sdds}(c) has offset value $6$. The sum of offset values
for this node is $1 + 6 = 7$ and $\idtov(D_{P_\alpha})$ corresponds to the leaf vtree node having variable $C$.
In this way, VS-SDD can be seen as an SDD variant that employs an indirect way of representing the respecting vtree nodes.


\subsection{Substitution-equivalency in VS-SDDs}
\label{sec:subeq}
Next we show how substitution-equivalent functions are shared in VS-SDD.
In Fig.~\ref{fig:sdds}(d), we should observe that the bottom-right node (say $\beta$) represents two substitution-equivalent functions $A\wedge B$ and $C\wedge D$. There are two different paths (say $P_1$ and $P_2$) from the root to $\beta$, and
they correspond to different offset values $D_{P_1} = 1 + 1 = 2$ and $D_{P_2} = 1 + 4 = 5$. Therefore, $\beta$ is used in two VS-SDDs $(\beta, 2)$ and $(\beta, 5)$ and 
they correspond to $A\wedge B$ and $C\wedge D$, respectively. 
In this way, substitution-equivalent functions are represented by VS-SDDs with the same structure and different offsets.

Now we proceed to the formal description.
Let $u, w$ be isomorphic subtrees of vtree $v$, $\mathbf{X}$ be the set of variables corresponding to the leaves of $v$, and $M$ be the number of variables.
We consider permutation $\pi_{u, w} : \{1, \ldots, M\} \mapsto \{1, \ldots, M\}$
that preserves the relation between $u$ and $w$. That is, 
let $X_i$ and $X_j$ be the variables associated with leaf nodes $u^\prime$ in $u$ and $w^\prime$ in $w$, respectively. We assume that $u^\prime$ and $w^\prime$ are associated through the graph isomorphism between $u$ and $w$. Then $\pi_{u, w}$ is the bijection satisfying $\pi_{u, w}(j) = i$ for every pair of $X_i$ and $X_j$ corresponding to the leaf nodes of $u$ and $w$.
If $u$ and $w$ are isomorphic and we employ preorder IDs, then the difference in IDs of corresponding nodes of $u$ and $w$ is unique. We call this the \emph{shift} between $u, w$ and denote it as $\delta$. For example, 
in the vtree in Fig.~\ref{fig:sdds}(a), two child nodes of the root node represent
isomorphic vtrees. In these vtrees $\delta = 3$ for every corresponding node pair.

\begin{theorem}
    Let $f, g$ be Boolean functions that essentially depend on isomorphic vtrees $u$ and $w$ (resp.), where $u$ and $w$ are nodes in the entire vtree $v$. 
    If $f$ and $g$ are
    substitution-equivalent with $\pi_{u, w}$ then the compressed and trimmed VS-SDDs $(\alpha, k)$ and $(\beta, \ell)$ representing $f$ and $g$ satisfies $\alpha = \beta$ and $\ell = k + \delta$.
\end{theorem}
\begin{proof}
    The Boolean function $\langle \alpha, k + \delta \rangle$ is the one wherein every appearance of every variable $l(\idtov(i))$ in $\langle \alpha, k\rangle$ is replaced with $l(\idtov(i+\delta))$.
    It is equivalent to $\langle \beta, \ell \rangle$.
\end{proof}

It is possible that there exist two VS-SDDs $(\alpha, k)$ and $(\beta, \ell)$ where
$\alpha = \beta$ but vtrees $\idtov(k)$ and $\idtov(\ell)$ are not isomorphic.
In such case, we do not share their structure.
In other words, we share the identical structures only when for the offsets $k$ and $\ell$, $\idtov(k)$ and $\idtov(\ell)$ are isomorphic.
We call this the \emph{identical vtree rule}. 
The VS-SDD in Fig.~\ref{fig:sdds}(d) satisfies this rule.
Such rule is unique to VS-SDDs, since in SDDs all identical structures are fully merged (i.e.~reduced).
This rule is crucial for guaranteeing some attractive properties of VS-SDDs introduced in later sections.

\section{Properties of VS-SDD}
\label{sec:property}
We show here some basic VS-SDD properties.
First, we prove the canonicity of some classes of VS-SDD.
Then we give proofs on VS-SDD size.

\subsection{Canonicity}
We say a class of VS-SDD  is canonical iff, given a vtree, for any Boolean function $f$, there is exactly one VS-SDD in this class representing $f$. 
We first introduce two classes of VS-SDDs, both have counterparts in SDDs.
\begin{definition}
    We say VS-SDD $(\alpha, k)$ is \emph{compressed} iff for each VS-SDD $(\beta, \ell)$ appearing in $(\alpha, k)$ where $\beta$ is a decomposition, it forms compressed $X$-partition.
    We say a VS-SDD $(\alpha, k)$ is \emph{trimmed} if it contains no decompositions with form of $\{([\top, \cdot], [\beta, d])\}$ and $\{
    ([\beta, d], [\top, \cdot]), ([\neg \beta, d], [\bot, \cdot])\}$. 
    We also say VS-SDD is \emph{lightly trimmed} if it contains no decompositions with form of $\{([\top, \cdot], [\top, \cdot])\}$ and $\{([\top, \cdot], [\bot, \cdot])\}$.
    We say VS-SDD $(\alpha, k)$ is \emph{normalized} iff for each VS-SDD $(\beta, \ell)$ appearing in $(\alpha, k)$ where $\beta$ is a decomposition, every prime $[p_i, d_i]$ ensures that $\idtov(d_i + \ell)$ is the left child of vtree node $\idtov(\ell)$ and every sub $[s_i, e_i]$ ensures that $\idtov( e_i + \ell)$ is the right child of vtree node $\idtov(\ell)$.
    \end{definition}
    
The proof of canonicity is almost identical to that for SDDs.
We first introduce some concepts and notations.
We use $(\alpha, k) \equiv (\beta, \ell)$ to represent that the corresponding Boolean functions are identical. 
\begin{definition}
    A Boolean function $f$ \emph{essentially depends} on vtree node $v$ if $f$ is not trivial and $f$ is a deepest node that includes all variables that $f$ essentially depends on. 
\end{definition}
\begin{lemma}[\cite{darwiche11sdd}]
A non-trivial function essentially depends on exactly one vtree node.
\end{lemma}
\begin{lemma}
\label{lem:essdep}
Let $(\alpha, k)$ be a trimmed and compressed VS-SDD. If $(\alpha, k) \equiv \mathfalse$, then $\alpha = \bot$. If $(\alpha,k) \equiv \mathtrue$, then $\alpha = \top$. Otherwise, $\idtov(k)$ always equals to the vtree node $v$ that $\langle\alpha,k\rangle$ essentially depends on.
\end{lemma}

The above lemma suggests that compressed and trimmed VS-SDDs can be partitioned into groups depending on the offset. We can prove the canonicity by exploiting this fact.
\begin{theorem}
	\label{thm:can-vs}
	Compressed and trimmed VS-SDDs with the same vtree, $v$, are canonical.
	Also, compressed, lightly trimmed, and normalized VS-SDDs with the same vtree,$v$, are canonical.
\end{theorem}
\begin{proof}
Here we give the proof for the case of compressed and trimmed VS-SDDs.
The proof for compressed, lightly trimmed and normalized SDDs can be
constructed in a similar way.

If two compressed SDDs $(\alpha, k)$ and $(\beta, \ell)$ satisfy $(\alpha, k) = (\beta, \ell)$, then $(\alpha, k) \equiv (\beta, \ell)$ from the definition.
Suppose $\langle \alpha, k\rangle  = \langle \beta, \ell\rangle$ and let $f = \langle \alpha, k \rangle = \langle \beta, \ell \rangle$. 
If $f = \mathtrue$, then $\alpha = \beta = \top$ and they are canonical. Similarly, if $f = \mathfalse$, then $\alpha = \beta = \bot$.

Next we consider the case of $f$ being non-trivial. From Lemma~\ref{lem:essdep}, $\idtov(k)=w=\idtov(\ell)$ where $w$ is the vtree node that $f$ essentially depends on.
Suppose $w$ is a leaf, then VS-SDDs must be literals and hence $(\alpha, k) = (\beta, \ell)$. Suppose now that $w$ is internal and that the theorem 
holds for VS-SDDs whose offsets correspond to descendant nodes of $\idtov(k)$. 
Let $w^l$ and $w^r$ be the left and the right subtree of $w$, respectively. Let $\mathbf{X}$ be variables in $w^l$, $\mathbf{Y}$  be variables in $w^r$, $\alpha = \{([p_1, d_1], [s_1, e_1]), \ldots, ([p_n, d_n], [s_n, e_n])\}$ and $\beta = \{([q_1, b_1], [r_1, c_1]), \ldots, ([q_m, b_m], [r_m, c_m])\}$.
By the definition, offsets $d_i + k$ and $b_j + \ell$ correspond to vtree nodes in $w^l$ and offsets $e_i + k$ and $c_j + \ell$ correspond to vtree nodes in $w^r$.
Since compressed $\mathbf{X}$-partitions $\{(\langle p_1, d_1 \rangle, \langle s_1, e_1 \rangle), \ldots, (\langle p_n, d_n \rangle,  \langle s_n, e_n \rangle)\}$ and
$\{(\langle q_1, b_1 \rangle, \langle r_1, c_1 \rangle), \ldots, (\langle q_m, b_m \rangle,  \langle r_m, c_m \rangle)\}$ are identical (see Theorem 3 of \cite{darwiche11sdd}), $n = m$ and there is a one-to-one $\equiv$-correspondence between the primes and subs.
From the inductive hypothesis, this means there is a one-to-one $=$-correspondence between the primes and subs. 
This implies $\alpha = \beta$ and thus $(\alpha, k) = (\beta, \ell)$.
\end{proof}

\subsection{About the Size: Exponential Compression}
\label{sec:exp}
We here compare VS-SDD size with SDD size.
First of all, we observe that VS-SDD is always smaller than SDDs since it is made by 
sharing substitution-equivalent nodes in SDDs and no other size changes occur.
\begin{proposition}
	\label{prop:less}
	For any SDD $\alpha$ defined with vtree $v$, there exists a VS-SDD whose size is not larger than $|\alpha|$.
\end{proposition}

We turn our focus to the best compression ratio of the VS-SDD. 
Since a vtree has $M$ leaves where $M$ is the number of variables, a vtree might have at most $M$ isomorphic subtrees.
Thus  the lower bound of  VS-SDD size  is $1/M$ of SDD when we employ the identical vtree rule.
Here we prove that there is a series of functions that almost achieves this compression ratio asymptotically.
\begin{theorem}
	\label{thm:ratio}
	There exists a sequence of Boolean functions $f_1,f_2,\ldots$ such that $f_j$ uses $\order{2^j}$ variables, the size of a compressed SDD representing $f_j$ is $\morder{2^j}$ with any vtree, and that of a compressed VS-SDD representing $f_j$ is $\order{j}$ with a particular vtree.
\end{theorem}

The compression ratio is $\order{j/2^j}=\order{\log M/M}$.
Theorem \ref{thm:ratio} makes a stronger statement, because ``any'' vtree can be considered for SDD.

One of the sequences satisfying Theorem \ref{thm:ratio} is as follows:
\fontsize{9pt}{0pt}\selectfont
\begin{align*}
	\textstyle
	f_j(\mathbf{X}) = (\neg X_1\vee\neg X_2)\wedge\bigwedge_{i=1}^{2^j-2}\bigl(\left(\neg X_i\vee\neg X_{2i+1}\right)\wedge\left(\neg X_{i}\vee\neg X_{2i+2}\right)
	\wedge\left(\neg X_{2i+1}\vee\neg X_{2i+2}\right)\bigr).
\end{align*}
\normalsize
By considering a complete binary tree like Fig.~\ref{fig:recursive}(a), we observe that $f_j(\mathbf{x})=\mathtrue$ iff the edges whose corresponding variables are set to $\mathtrue$ constitute a matching.

\begin{figure}[tbp]
	\centering
	\includegraphics[keepaspectratio,scale=0.8]{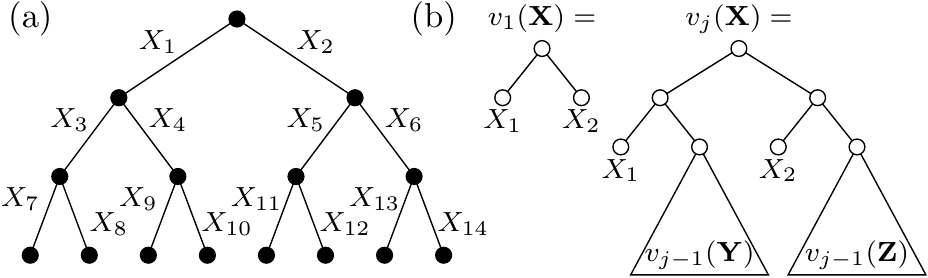}
	\caption{(a) A complete binary tree with variable-labeled edges. (b) The recursive structure of vtree $v_j(\mathbf{X})$. Here $X_1$ and $X_2$ indicate the first and second (resp.) variables of $\mathbf{X}$.}
	\label{fig:recursive}
\end{figure}

We outline the proof here; details are given in the Appendix.
Since the first part of Theorem \ref{thm:ratio} can easily be proved, we refer to the second part.
We define vtree $v_j(\mathbf{X})$ in a recursive manner as shown in Fig.~\ref{fig:recursive}(b).
Here $\mathbf{Y}$ includes the variables corresponding to the edges below $X_1$ when considering the complete binary tree as in Fig.~\ref{fig:recursive}(a) (namely $X_3, X_4, \ldots$) and $\mathbf{Z}$ includes those corresponding to the edges below $X_2$ ($X_5, X_6, \ldots$).
Now we decompose $f_j(\mathbf{X})$ with respect to $v_j(\mathbf{X})$ by using $f_{j-1}(\mathbf{Y})$, $f_{j-1}(\mathbf{Z})$ and some other subfunctions.
We then use the fact that $f_{j-1}(\mathbf{Y})$ and $f_{j-1}(\mathbf{Z})$ are substitution-equivalent with $\pi_{v_{j-1}(\mathbf{Y}),v_{j-1}(\mathbf{Z})}$.
By repetitively applying this argument, we observe that by decomposing $f_j$ with respect to $v_j$, the SDD of $f_j$ has $2^i$ nodes that represent $f_{j-i}(\cdot)$, which all represent substitution-equivalent functions, and thus the VS-SDD reduces the size exponentially.

\section{Operations of VS-SDD}
\label{sec:operation}

The most important property of VS-SDDs is that they support numerous key operations in polytime.
We focus here on the important queries and transformations shown in \cite{darwiche02KC}.
Most of these operations are based on $\mathapply$.
$\mathapply$ takes, given a vtree, two VS-SDDs $(\alpha, k)$, $(\beta, \ell)$  and binary operation $\circ$ such as $\vee$ (disjunction), $\wedge$ (conjunction) and $\oplus$ (exclusive-or), and returns a VS-SDD of $\langle\alpha,k\rangle\circ\langle\beta, \ell \rangle$.
By repeating $\mathapply$ operations, we can flexibly construct VS-SDDs representing various Boolean functions.

To simplify the explanation of $\mathapply$, we assume that VS-SDDs are normalized and thus have the same offset value $k$.
Given two normalized VS-SDDs $(\alpha, k)$, $(\beta, k)$,
Alg.~\ref{alg:apply-norm} provides pseudocode for the function $\mathapply(\alpha, \beta, k, \circ)$.
The mechanism behind the $\mathapply$ computation of VS-SDDs is as follows.
Let $f,g$ be Boolean functions with the same variable set, and suppose that $f$ is $\mathbf{X}$-partitioned as $f=\bigvee_{i=1}^{n}[p_i(\mathbf{X})\wedge s_i(\mathbf{Y})]$ and $g$ is also $\mathbf{X}$-partitioned (with the same $\mathbf{X}$) as $g=\bigvee_{j=1}^{m}[q_j(\mathbf{X})\wedge r_j(\mathbf{Y})]$.
Then, $f\circ g$ can be expressed as $\bigvee_{i=1}^{n}\bigvee_{j=1}^{m}[(p_i(\mathbf{X})\wedge q_j(\mathbf{X}))\wedge(s_i(\mathbf{Y})\circ r_j(\mathbf{Y}))]$,
 where $(p_i(\mathbf{X})\wedge q_j(\mathbf{X}))\wedge(p_{i'}(\mathbf{X})\wedge q_{j'}(\mathbf{X}))=\mathfalse$ for $(i,j)\neq (i',j')$ and $\bigvee_{i=1}^{n}\bigvee_{j=1}^{m}(p_i(\mathbf{X})\wedge q_j(\mathbf{X}))=\mathtrue$.
Thus, computing $p_i\wedge q_j$ and $s_i\circ r_j$ for each $(i,j)$ pair and ignoring the pairs such that $p_i\wedge q_j=\mathfalse$ yields the $\mathbf{X}$-partition of $f\circ g$.
Alg.~\ref{alg:apply-norm} follows this recursive definition.  

\begin{algorithm}[tbp]
	\caption{$\mathtt{Apply}(\alpha,\beta,k,\circ)$, which computes a VS-SDD representing $\langle\alpha, k\rangle\circ\langle\beta, k\rangle$ for two normalized VS-SDDs $(\alpha, k),(\beta, k)$ and a binary operator $\circ$.}
	\label{alg:apply-norm}
	\fontsize{9pt}{7pt}\selectfont
	$\mathtt{Cache}(\cdot, \cdot, \cdot) = \mathtt{nil}$ initially. $\mathtt{Expand}(\alpha)$ returns $\{([\top, \cdot], [\top, \cdot])\}$ if $\alpha = \top$; $\{([\top, \cdot], [\bot, \cdot])\}$ if $\alpha = \bot$; else $\alpha$. $\mathtt{UniqueD}(\gamma)$ returns $\top$ if $\gamma = \{([\top, \cdot],[\top, \cdot])\}$; $\bot$ if $\gamma = \{([\top, \cdot], [\bot, \cdot])\}$; else the unique VS-SDD with elements $\gamma$.
	
	\begin{algorithmic}[1]
		\IF{$\alpha$ and $\beta$ are either of $\top, \bot, \mathbf{v}, \neg \mathbf{v}$} 
        \STATE \textbf{return} the pair of corresponding value and offset.
		\ELSIF{ $\mathtt{Cache}(\alpha, \beta, \circ) \neq \mathtt{nil}$}
		\STATE  $\lambda \gets \mathtt{Cache}(\alpha, \beta, \circ)$
		\STATE \textbf{return} $(\lambda, k)$
		\ELSE
    	\STATE $\gamma \gets \{\}$
    	\FORALL{elements $([p_i, d], [s_i, e])$ in $\mathtt{Expand}(\alpha)$}
    	\FORALL{elements $([q_j, d], [r_j, e])$ in $\mathtt{Expand}(\beta)$}
    	\STATE $(p, \ell_p) \gets \mathtt{Apply}(p_i, q_j, d+k, \circ)$
        \IF{$(p, \ell_p)$ is consistent}
        \STATE $(s, \ell_s) \gets \mathtt{Apply}(s_i, r_j, e+k, \circ)$
        \STATE add element $([p, \ell_p - k], [s, \ell_s - k])$ to $\gamma$
        \ENDIF
    	\ENDFOR
    	\ENDFOR
		\ENDIF
		\STATE  $\lambda \gets \mathtt{UniqueD}(\gamma)$, $\mathtt{Cache}(\alpha, \beta, \circ) \gets \lambda$
		\STATE \textbf{return} $(\lambda, k)$
	\normalsize
	\end{algorithmic}
\end{algorithm}

\begin{proposition}
    \label{prop:apply}
    $\mathapply (\alpha, \beta, k, \circ)$ runs in $O(|\alpha||\beta|)$ time.
\end{proposition}

The above result  is the same as in the case of $\mathapply$ for SDDs. The key to achieving this result is that we use $\mathtt{Cache}$ without using offset $k$ as a key. We use the fact that if a pair of functions $f(\mathbf{X})$, $f^\prime(\mathbf{Y})$ and $g(\mathbf{X})$, $g^\prime(\mathbf{Y})$, where $\mathbf{X}$ and $\mathbf{Y}$ are non-overlapping, are substitution-equivalent with permutation $\pi$,
then the composed functions $f\circ g$ and $f^\prime \circ g^\prime$ are also 
substitution-equivalent with the same permutation $\pi$.
For example, let $f=A\wedge B$, $f'=C\wedge D$, $g=\neg A$ and $g'=\neg C$, in which $f$ and $f'$, and $g$ and $g'$ (resp.) are substitution-equivalent with permutation $\pi_{\idtov(2), \idtov(5)}$ defined with the vtree in Fig.~\ref{fig:sdds}(a).
Then $f\vee g=\neg A\vee B$ and $f'\vee g'=\neg C\vee D$ are also substitution-equivalent with $\pi_{\idtov(2), \idtov(5)}$. This means the results of $\mathapply(\alpha, \beta, k, \circ)$ with different $k$ are all substitution-equivalent. Therefore, we can reuse the result
obtained with different offsets.

If VS-SDDs are trimmed, we can also define $\mathapply$ operations
for them. While similar to the case for trimmed SDDs, the $\mathapply$ for trimmed VS-SDDs
are more complicated than that of normalized VS-SDDs since we have to take
different operations depending on the combination of offset values of input VS-SDDs.
However, the complexity of $\mathapply$ for trimmed VS-SDDs is also $O(|\alpha||\beta|)$.
We detail the $\mathapply$ for trimmed VS-SDDs in the Appendix \ref{appx:trimmedapply}.

Note that even if two VS-SDDs are compressed, the resulting VS-SDD cannot be assumed to be compressed since the same sub may appear.
It is said in \cite{broeck15comp} that there is a case in which compression makes an SDD exponentially larger, and thus a similar statement holds for VS-SDDs.
Therefore, if we oblige the output to be compressed, Prop.~\ref{prop:apply} does not hold.
Note that during $\mathapply$, compression can be performed by taking the disjunction of primes when the same subs emerge.

By extensively using Prop.~\ref{prop:apply} with some other algorithms, it can be shown that the various important queries and transformations in \cite{darwiche02KC} can be performed in polytime.
The proof is in the Appendix.

\setlength{\tabcolsep}{0.7mm}
\begin{table}[tbp]
    \caption{List of supported (a) queries and (b) transformations for SDDs (S), VS-SDDs (V), compressed SDDs (S(C)) and compressed VS-SDDs (V(C)).
    \checkmark\ indicates the existence of a polytime algorithm, while \textbullet\ indicates such polytime algorithm is shown to be impossible.}
    \label{tb:query}
	\begin{minipage}[t]{0.49\hsize}
		\centering
		\fontsize{9pt}{6pt}\selectfont
		\begin{tabular}{cccccc}
			\hline
			\multicolumn{2}{c}{(a) Query} & S & \textbf{V} & S(C) & \textbf{V(C)} \\
			\hline
			\textbf{CO} & consistency & \checkmark & \checkmark & \checkmark & \checkmark \\
			\textbf{VA} & validity & \checkmark & \checkmark & \checkmark & \checkmark \\
			\textbf{CE} & clausal entailment & \checkmark & \checkmark & \checkmark & \checkmark \\
			\textbf{IM} & implicant check & \checkmark & \checkmark & \checkmark & \checkmark \\
			\textbf{EQ} & equivalence check & \checkmark & \checkmark & \checkmark & \checkmark \\
			\textbf{CT} & model counting & \checkmark & \checkmark & \checkmark & \checkmark \\
			\textbf{SE} & sentential entailment & \checkmark & \checkmark & \checkmark & \checkmark \\
			\textbf{ME} & model enumeration & \checkmark & \checkmark & \checkmark & \checkmark \\
			\hline
		\end{tabular}
		\normalsize
	\end{minipage}
	\begin{minipage}[t]{0.5\hsize}
		\centering
		\fontsize{7pt}{6pt}\selectfont
		\begin{tabular}{cccccc}
			\hline
			\multicolumn{2}{c}{(b) Transformation} & S & \textbf{V} & S(C) & \textbf{V(C)} \\
			\hline
			$\mathbf{\wedge}$\textbf{C} & conjunction & \textbullet & \textbullet & \textbullet & \textbullet \\
			$\mathbf{\wedge}$\textbf{BC} & bounded conjunction & \checkmark & \checkmark & \textbullet & \textbullet \\
			$\mathbf{\vee}$\textbf{C} & disjunction & \textbullet & \textbullet & \textbullet & \textbullet \\
			$\mathbf{\vee}$\textbf{BC} & bounded disjunction & \checkmark & \checkmark & \textbullet & \textbullet \\
			$\mathbf{\neg}$\textbf{C} & negation & \checkmark & \checkmark & \checkmark & \checkmark \\
			\textbf{CD} & conditioning & \checkmark & \checkmark & \textbullet & \textbullet \\
			\textbf{FO} & forgetting & \textbullet & \textbullet & \textbullet & \textbullet \\
			\textbf{SFO} & singleton forgetting & \checkmark & \checkmark & \textbullet & \textbullet \\
			\hline
		\end{tabular}
		\normalsize
	\end{minipage}
\end{table}
\begin{proposition}
	\label{prop:query}
	The results in Table \ref{tb:query} hold.
\end{proposition}

	Note that some applications, e.g.~probabilistic inference~\cite{sang05weight,darwiche02KC}, need \emph{weighted} model counting, where each variable has a weight.
	Though this cannot be performed in $\order{|\alpha|}$ time for VS-SDD $\alpha$, it can be performed at least as fast as is possible by using the corresponding SDD, by preparing, for each node, as many counters as the number of unified nodes in the original SDD.
	Moreover, if the weights of variables are the same for the same vtree structures, we can share counters, which speeds up the computation.

\section{Implementation}
\label{sec:implement}
We should address implementation in order to ensure space-efficiency.
One suspects that even if VS-SDD size is never larger than SDD size, the memory usage may increase because we store the information of respecting vtree node ids in the edges of a diagram (differentially) instead of in the nodes.
This is true if VS-SDDs are implemented as is.

However, a small modification avoids this problem.
First, for a normalized VS-SDD, we simply ignore the differences of vtree node ids attached to the edges.
Even so, we can recover the respecting vtree node because if an SDD node respects vtree $v$, its primes respect the left child of $v$ and its subs respect the right child of $v$.
Second, for a general VS-SDD, we just reuse the structure of the original SDD.
Among the SDD nodes that are merged into one in the VS-SDD structure, we just leave one representative (e.g.~the one with the smallest respecting vtree node id).
Then each of the other nodes has a pointer to the representative node instead of storing the prime-sub pairs.
An example of such a structure is drawn in Fig.~\ref{fig:vssddimpl}(c).
Here the dashed arrow indicates a pointer to the representative node described above.
Since each decomposition node has at least one prime-sub pair that typically uses two pointers, replacing it by single pointer will never increase memory usage.
Working with such a structure does not violate any properties about VS-SDDs, including the operations described above.

\begin{figure}[tbp]
	\centering
	\includegraphics[keepaspectratio,scale=0.75]{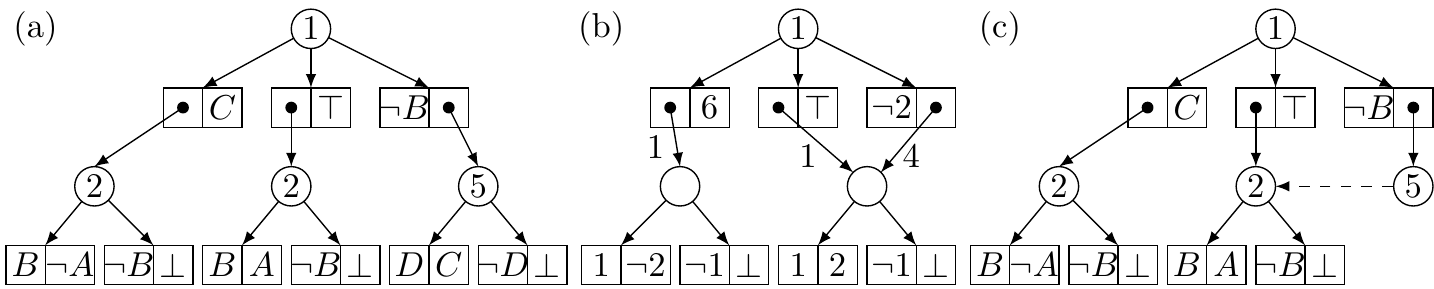}
	\caption{(a)(b) An SDD and a VS-SDD that are the same as Fig.~\ref{fig:sdds}. (c) The example of the representation of VS-SDD (b) using original SDD structure (a). The dashed arrow indicates a pointer to the representative node.}
	\label{fig:vssddimpl}
\end{figure}

\section{Evaluation}
\label{sec:eval}
We use some benchmarks of Boolean functions to evaluate how our approach reduces the size of an SDD.
we compile a CNF into an SDD with the dynamic vtree search~\cite{choi13dynamic} and then compare the sizes yielded by the SDD and its VS form (VS-SDD).
To compile a CNF, we use the SDD package version 2.0~\cite{choi18sdd} with a balanced initial vtree.
Here note that we use for both SDD and VS-SDD the same vtree, which is searched to suit for SDD.
All the experiments are conducted on a 64-bit macOS (High Sierra) machine with 2.5 GHz Intel Core i7 CPU ($1$ thread) and 16 GB RAM.

Here we focus on the planning CNF dataset that was used in the experiment of Sym-DDG~\cite{bart14sym}. 
The planning problem naturally exhibits symmetries, e.g.~see \cite{palacios05plan}.
Given time horizon $T$, this data represents a deterministic planning problem with varying initial and goal states.
Here we can choose an action from a fixed action set for each time point, and a plan for this problem is a time series of actions for $t=0,\ldots,T-1$ that leads from the initial state to the goal state.
For more details, see \cite{bart14sym}. 
We use the planning problems that were also used in the experiment of Sym-DDG: ``blocks-2'', ``bomb-5-1'', ``comm-5-2'', ``emptyroom-4/8'', and ``safe-5/30'', with varying time horizons $T=3,5,7,10$.

The next focus is on the benchmarks with apparent symmetries.
The first one is the $N$-queens problem, that is, given an $N\times N$ chessboard, place $N$ queens such that no two queens attack each other.
We assign a variable to each square in the chessboard, and consider a Boolean function that evaluates $\mathtrue$ iff the $\mathtrue$ variables constitute one answer for this problem.
This problem is used as a benchmark in Zero-suppressed BDD and other DD studies~\cite{minato93zdd,bryant18chain}.
The second one is enumerating matchings of grid graphs.
Subgraph enumeration with decision diagrams has several applications; see \cite{knuth11TACP} and \cite{nishino17graph}.
Here the grid graph is often used as a benchmark~\cite{iwashita13SAP}, because it is closely related to self-avoiding walk~\cite{madras96SAW}, and subgraph enumeration becomes much harder for larger grids despite their simplicity.
We can observe that both the chessboard and the grid have line symmetries and point symmetry. Again we exploit dynamic vtree search implemented in the SDD package.

\setlength{\tabcolsep}{0.5mm}
\begin{table}
	\centering
	\caption{Results for experiments. The ``S'' column represents SDD size, ``V'' represents VS-SDD size, and ``ratio'' indicates the ratio of VS-SDD size compared to SDD size.}
	\label{tb:eval2}
	\fontsize{7.5pt}{6pt}\selectfont
	\begin{tabular}{cr|rrr}
		\hline
		Problem & \multicolumn{1}{c|}{\#vars} & \multicolumn{1}{c}{S} & \multicolumn{1}{c}{\textbf{V}} & \multicolumn{1}{c}{ratio}\\
		\hline
		blocks-2\textunderscore t3 & 248 & 8811 & 7057 & 80.1\% \\
		blocks-2\textunderscore t5 & 406 & 31861 & 28858 & 90.6\% \\
		\hline
		bomb-5-1\textunderscore t3 & 348 & 3798 & 2278 & 60.0\% \\
		bomb-5-1\textunderscore t5 & 564 & 6327 & 3960 & 62.6\% \\
		bomb-5-1\textunderscore t7 & 780 & 11212 & 7287 & 65.0\% \\
		bomb-5-1\textunderscore t10 & 1104 & 16514 & 10426 & 63.1\% \\
		\hline
		comm-5-2\textunderscore t3 & 488 & 20584 & 18033 & 87.6\% \\
		\hline
		emptyroom-4\textunderscore t3 & 116 & 1822 & 1146 & 62.9\% \\
		emptyroom-4\textunderscore t5 & 188 & 3090 & 1885 & 61.0\% \\
		emptyroom-4\textunderscore t7 & 260 & 5073 & 3001 & 59.2\% \\
		emptyroom-4\textunderscore t10 & 368 & 106737 & 103417 & 96.8\% \\
		\hline
		emptyroom-8\textunderscore t3 & 244 & 10511 & 8549 & 81.3\% \\
		\hline
		safe-5\textunderscore t3 & 54 & 567 & 441 & 77.8\% \\
		safe-5\textunderscore t5 & 86 & 898 & 640 & 71.2\% \\
		safe-5\textunderscore t7 & 118 & 1710 & 1314 & 76.8\% \\
		safe-5\textunderscore t10 & 166 & 2506 & 1756 & 70.1\% \\
		\hline
		safe-30\textunderscore t3 & 304 & 5476 & 4067 & 74.3\% \\
		safe-30\textunderscore t5 & 486 & 8710 & 6328 & 72.7\% \\
		safe-30\textunderscore t7 & 668 & 14449 & 10371 & 71.8\% \\
		safe-30\textunderscore t10 & 941 & 23469 & 17421 & 74.2\% \\
		\hline
		8-Queens & 64 & 2222 & 1624 & 73.1\% \\
		9-Queens & 81 & 5559 & 4767 & 85.8\% \\
		10-Queens & 100 & 10351 & 9159 & 88.5\% \\
		11-Queens & 121 & 30611 & 28876 & 94.3\% \\
		\hline
		Matching-6x6 & 60 & 13091 & 12671 & 96.8\% \\
		Matching-8x8 & 112 & 98200 & 97103 & 98.8\% \\
		Matching-6x18 & 192 & 36228 & 34241 & 94.5\% \\
		\hline
	\end{tabular}
\end{table}

Table \ref{tb:eval2} shows the results of our experiments.
The ``S'' column represents SDD size, ``V'' represents VS-SDD size, and ``ratio'' indicates the ratio of VS-SDD size compared to SDD size.
Here the problems in which the SDD compilation took more than 10 minutes are omitted.
For planning problems, the suffix ``\textunderscore t$n$'' stands for $T=n$, and for matching problems, the suffix indicates the grid size.
It is observed that for many planning problems, the VS-SDD reduces the size to around 60\% to 80\% of the original SDD.
We observe that for these cases, many nodes representing substitution-equivalent functions are found among the bottom nodes of the original SDD, which yields the substantial size decrease.
These compression ratios are competitive to, and for some cases better than, that of the Sym-DDG~\cite{bart14sym} compared to the DDG.
For the $N$-queens problems, still better compression ratios are achieved except for $N=11$.
However, for matching enumeration problems, the effect of variable shift is relatively small.
One reason is the asymmetry of primes and subs, that is, primes must form a partition while subs do not have such a limitation.
The success in planning datasets may be explained as follows.
The dynamic vtree search typically gathers variables with strong dependence locally to achieve succinctness.
For planning problems, the variables with near time points are gathered, which captures the symmetric nature of the problem.

\section{Conclusion}
\label{sec:conclusion}
We proposed a variable shift SDD (VS-SDD), a more succinct variant of SDD that is obtained by changing the way in which respecting vtree nodes are indicated.
VS-SDD keeps the two important properties of SDDs, the canonicity and the support of many useful operations.
The size of a VS-SDD is always smaller than or equal to that of an SDD, and there are cases where the VS-SDD is exponentially smaller than the SDD.
Experiments show that our idea effectively captures the symmetries of Boolean functions, which leads to succinct compilation.

\bibliography{references}

\begin{thebibliography}{10}

\bibitem{anuchi95diff}
Anuchit Anuchitanukul, Zohar Manna, and Tom{\'a}s~E. Uribe.
\newblock Differential {BDDs}.
\newblock In {\em Computer Science Today}, pages 218--233, 1995.

\bibitem{bart14sym}
Anicet Bart, Fr\'ed\'eric Koriche, {Jean-Marie} Lagniez, and Pierre Marquis.
\newblock Symmetry-driven decision diagrams for knowledge compilation.
\newblock In {\em ECAI}, pages 51--56, 2014.

\bibitem{bova16exp}
Simone Bova.
\newblock {SDDs} are exponentially more succinct than {OBDDs}.
\newblock In {\em AAAI}, pages 929--935, 2016.

\bibitem{broeck15comp}
Guy~{Van den} Broeck and Adnan Darwiche.
\newblock On the role of canonicity in knowledge compilation.
\newblock In {\em AAAI}, pages 1641--1648, 2015.

\bibitem{bryant86bdd}
Randal~E. Bryant.
\newblock Graph-based algorithms for boolean function manipulation.
\newblock {\em IEEE Trans. Comput.}, C-35:677--691, 1986.

\bibitem{bryant18chain}
Randal~E. Bryant.
\newblock Chain reduction for binary and zero-suppressed decision diagrams.
\newblock In {\em TACAS}, pages 81--98, 2018.

\bibitem{choi13dynamic}
Arthur Choi and Adnan Darwiche.
\newblock Dynamic minimization of sentential decision diagrams.
\newblock In {\em AAAI}, pages 187--194, 2013.

\bibitem{choi18sdd}
Arthur Choi and Adnan Darwiche.
\newblock The {SDD} package: version 2.0.
\newblock \url{http://reasoning.cs.ucla.edu/sdd/}, 2018.

\bibitem{darwiche11sdd}
Adnan Darwiche.
\newblock {SDD}: a new canonical representation of propositional knowledge
  bases.
\newblock In {\em AAAI}, pages 819--826, 2011.

\bibitem{darwiche02KC}
Adnan Darwiche and Pierre Marquis.
\newblock A knowledge compilation map.
\newblock {\em J. Artif. Intell. Res.}, 17:229--264, 2002.

\bibitem{duenas17count}
Leonardo {Duenas-Osorio}, Kuldeep~S. Meel, Roger Paredes, and Moshe~Y. Vardi.
\newblock Counting-based reliability estimation for power-transmission.
\newblock In {\em AAAI}, pages 4488--4494, 2017.

\bibitem{fragier06DDG}
H\'ei\`ene Fragier and Pierre Marquis.
\newblock On the use of partially ordered decision graphs for knowledge
  compilation and quantified {Boolean} formulae.
\newblock In {\em AAAI}, pages 42--47, 2006.

\bibitem{gergov94FBDD}
Jordan Gergov and Christoph Meinel.
\newblock Efficient {Boolean} manipulation with {OBDD's} can be extended to
  {FBDD's}.
\newblock {\em IEEE Trans. Comput.}, 43:1197--1209, 1994.

\bibitem{harel84LCA}
Dov Harel and Robert~E. Tarjan.
\newblock Fast algorithms for finding nearest common ancestors.
\newblock {\em SIAM J. Comput.}, 13:338--355, 1984.

\bibitem{iwashita13SAP}
Hiroaki Iwashita, Yoshio Nakazawa, Jun Kawahara, Takeaki Uno, and {Shin-ichi}
  Minato.
\newblock Efficient computation of the number of paths in a grid graph with
  minimal perfect hash functions.
\newblock Technical Report TCS-TR-A-13-64, Division of Computer Science,
  Hokkaido University, 2013.

\bibitem{knuth11TACP}
Donald~E. Knuth.
\newblock {\em The Art of Computer Programming}, volume 4A: Combinatorial
  Algorithms, Part I.
\newblock Addison-Wesley, 2011.

\bibitem{madras96SAW}
Neal Madras and Gordon Slade.
\newblock {\em The Self-Avoiding Walk}.
\newblock Birkh{\"a}user Basel, 2011.

\bibitem{madre88ce}
{Jean-Christophe} Madre and {Jean-Paul} Billon.
\newblock Proving circuit correctness using formal comparison between expected
  and extracted behaviour.
\newblock In {\em DAC}, pages 205--210, 1988.

\bibitem{minato93zdd}
{Shin-ichi} Minato.
\newblock Zero-suppressed {BDDs} for set manipulation in combinatorial
  problems.
\newblock In {\em DAC}, pages 272--277, 1993.

\bibitem{minato90vs}
{Shin-ichi} Minato, Nagisa Ishiura, and Shuzo Yajima.
\newblock Shared binary decision diagram with attributed edges for efficient
  boolean function manipulation.
\newblock In {\em DAC}, pages 52--57, 1990.

\bibitem{nishino17graph}
Masaaki Nishino, Norihito Yasuda, {Shin-ichi} Minato, and Masaaki Nagata.
\newblock Compiling graph substructures into sentential decision diagrams.
\newblock In {\em AAAI}, pages 1213--1221, 2017.

\bibitem{oztok15topdown}
Umut Oztok and Adnan Darwiche.
\newblock A top-down compiler for sentential decision diagrams.
\newblock In {\em IJCAI}, pages 3141--3148, 2015.

\bibitem{palacios05plan}
H\'ector Palacios, Blai Bonet, Adnan Darwiche, and H\'ector Geffner.
\newblock Pruning conformant plans by counting models on compiled {d-DNNF}
  representations.
\newblock In {\em ICAPS}, pages 141--150, 2005.

\bibitem{pipat08struct}
Knot Pipatsrisawat and Adnan Darwiche.
\newblock New compilation languages based on structured decomposability.
\newblock In {\em AAAI}, pages 517--522, 2008.

\bibitem{sang05weight}
Tian Sang, Paul Beame, and Henry Kautz.
\newblock Performing {Bayesian} inference by weighted model counting.
\newblock In {\em AAAI}, pages 475--481, 2005.

\bibitem{vlasse14LP}
Jonas Vlasselaer, Joris Renkens, Guy~{Van den} Broeck, and Luc~De Raedt.
\newblock Compiling probabilistic logic programs into sentential decision
  diagrams.
\newblock In {\em PLP}, 2014.

\end{thebibliography}

\appendix

\section{Appendix: Detailed Proofs}
\label{appx:proofs}

\begin{proof}[Proof of Theorem \ref{thm:ratio}]
	The first part can be proved by the following general claim.
	\begin{claim}
		\label{lem:SDDlower}
		Let $f(\mathbf{X})$ be a Boolean function such that for any variable $X_i$ in $\mathbf{X}$, the conditioned functions $f|_{X_i=\mathtrue}$ and $f|_{X_i=\mathfalse}$ are different.
		Then any SDD representing $f$ has size $\morder{M}$, where $M$ is the number of variables in $\mathbf{X}$.
		This holds for any vtree.
	\end{claim}
	\begin{claimproof}
		We prove that for any variable $X_i$ in $\mathbf{X}$, the SDD of $f$ contains at least either $X_i$ or $\neg X_i$ (as a literal SDD).
		If this holds, there are at least $m$ literals in the SDD of $f$ and thus at least $\lceil M/2\rceil$ prime-sub pairs, which suggests that its size is at least $\lceil M/2\rceil=\morder{M}$.
		
		Suppose the SDD of $f$ does not contain $X_i$ and $\neg X_i$.
		Recall the definition (semantics) of SDD.
		By recursively applying the definition of $\langle\cdot\rangle$, we obtain an expression of $f$ by using conjunctions, disjunctions, and literals.
		If the SDD does not contain $X_i$ and $\neg X_i$, this expression also does not contain $X_i$.
		This means that $f|_{X_i=\mathtrue}$ and $f|_{X_i=\mathfalse}$ are equivalent, since the assignment of $X_i$ is not mentioned in $f$.
		This contradicts the condition, thus the SDD contains at least one of $X_i$ and $\neg X_i$.
	\end{claimproof}

  Now we refer to the second part.
  We use the vtree $v_j(\mathbf{X})$ explained in the main article and drawn in Fig.~\ref{fig:recursive}(b).
  We give a proof by inductively showing that functions $f_j(\mathbf{X}), \neg f_j(\mathbf{X})$,  $f'_j(\mathbf{X})$ and $\neg f'_j(\mathbf{X})$ can be represented with $\order j$ size VS-SDDs, where $f'_j(\mathbf{X}) = \neg X_1 \wedge \neg X_2 \wedge f_j(\mathbf{X})$.

If $j \leq 2$, then the size  of VS-SDDs representing $f_j(\mathbf{X}), \neg f_j(\mathbf{X}), f'_j(\mathbf{X})$ and $\neg f'_j(\mathbf{X})$ are constant.
If $j \geq 3$, then the $(\{X_1\} \cup \mathbf{Y})$-partition of $f_j(\mathbf{X})$ 
 defined with vtree node $v_j(\mathbf{X})$ becomes 
		\fontsize{9pt}{0pt}\selectfont
		\begin{align*}
			& \Bigl(\bigl(X_1\wedge f'_{j-1}(\mathbf{Y})\bigr) \wedge \bigl(\neg X_2\wedge f_{j-1}(\mathbf{Z})\bigr)\Bigr)\vee \\
			& \Bigl(\bigl(\neg X_1\wedge f_{j-1}(\mathbf{Y})\bigr) \wedge \bigl((X_2\wedge f'_{j-1}(\mathbf{Z}))\vee (\neg X_2\wedge f_{j-1}(\mathbf{Z}))\bigr)\Bigr)\vee \\
			& \Bigl(\bigl((X_1\wedge\neg f'_{j-1}(\mathbf{Y}))\vee (\neg X_1\wedge\neg f_{j-1}(\mathbf{Y}))\bigr) \wedge\mathfalse\Bigr).
		\end{align*}
        \normalsize
The prime of the first element $X_1\wedge f'_{j-1}(\mathbf{Y})$ becomes $\mathtrue$ iff
the corresponding edges form a matching in the left half of the binary tree having edges $X_1, X_3, X_4, \ldots$ and $X_1 = \mathtrue$. The prime of the second element $\neg X_1 \wedge f_{j-1}(\mathbf{Y})$ becomes true
 iff  the selected edges form a matching in the left half tree and $X_1 = \mathfalse$. The prime of the third element
 becomes $\mathtrue$ iff the selected edges do not form a matching in the left half tree.
The above  partition is compressed since $f_{j-1}(\mathbf{Z})\neq\mathfalse$ and $f'_{j-1}(\mathbf{Z})\neq\mathfalse$.
		If we were to depict this decomposition as an SDD, it would look like Fig.~\ref{fig:recursive2}.
		The point is that pairs of $f_{j-1}(\mathbf{Y})$ and $f_{j-1}(\mathbf{Z})$, and $f'_{j-1}(\mathbf{Y})$ and $f'_{j-1}(\mathbf{Z})$ are substitution-equivalent  with $\pi_{v_{j-1}(\mathbf{Y}),v_{j-1}(\mathbf{Z})}$, and thus the VS-SDD representation prepares only one node for $f_{j-1}$ and $f'_{j-1}$ (see Fig.~\ref{fig:recursive2}).
    Therefore, the size of the VS-SDD representing $f_{j}(\mathbf{X})$ equals the sum of sizes of VS-SDDs representing $f_{j-1}(\mathbf{Y})$, $f'_{j-1}(\mathbf{Y})$, and $\neg f'_{j-1}(\mathbf{Y})$ plus a constant.
		Similarly, $f'_j(\mathbf{X})$ can be decomposed as 
		\fontsize{9pt}{0pt}\selectfont
		\begin{align*}
			\Bigl(\bigl(\neg X_1\wedge f_{j-1}(\mathbf{Y})\bigr)\wedge\bigl(\neg X_2\wedge f_{j-1}(\mathbf{Z})\bigr)\Bigr)\vee\Bigl(\bigl((X_1\wedge\mathtrue)\vee(\neg X_1\wedge\neg f_{j-1}(\mathbf{Y}))\bigr)\wedge\mathfalse\Bigr).
		\end{align*}
		\normalsize
		Here $f_{j-1}(\mathbf{Y})$ and $f_{j-1}(\mathbf{Z})$ are substitution-equivalent with $\pi_{v_{j-1}(\mathbf{Y}),v_{j-1}(\mathbf{Z})}$ and thus the size of a VS-SDD representing 
		$f'_{i}(\mathbf{X})$ equals the size of the VS-SDD representing $f_{j-1}(\mathbf{Y})$ plus a constant.
        The $(\{X_1\} \cup \mathbf{Y})$-partitions of $\neg f_j(\mathbf{X})$ and $\neg f'_i(\mathbf{X})$ are represented in almost the same way. Therefore, from the inductive hypothesis, VS-SDDs representing $f_j(\mathbf{X}), f'_j(\mathbf{X}), \neg f_j(\mathbf{X})$ and $\neg f'_j(\mathbf{X})$ have $\order j$ size.
\end{proof}

\begin{figure}[tbp]
	\centering
	\includegraphics[keepaspectratio,scale=0.7]{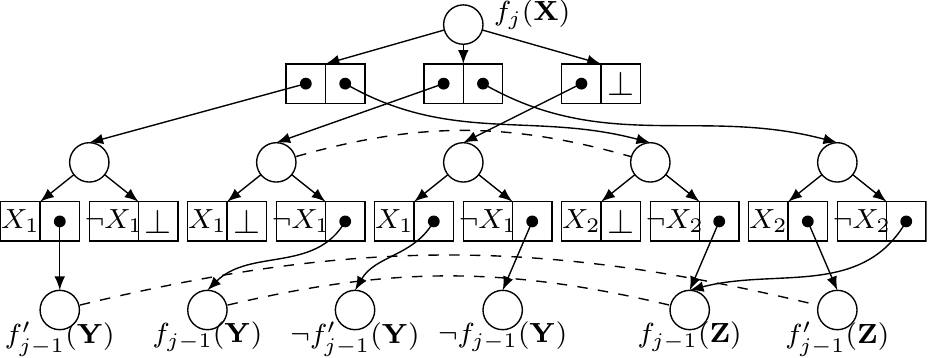}
	\caption{The recursive structure of the compressed SDD of $f_j(\mathbf{X})$ respecting $v_j(\mathbf{X})$ in Fig.~\ref{fig:recursive}(b). Dashed lines indicate that the nodes on both ends represents substitution-equivalent functions. Here $X_1$ and $X_2$ indicate the first and second (resp.) variables of $\mathbf{X}$.}
	\label{fig:recursive2}
\end{figure}

\begin{proof}[Proof of Prop.~\ref{prop:query}]
	First, we consider the queries in Table \ref{tb:query}(a).
	The first concern is the polytime solvability of model counting, i.e.~\textbf{CT}.
	The \emph{model count} of Boolean function $f$ is the number of satisfying assignments of $f$.
	Model counting, also known as \#SAT, is applicable to wider research areas, e.g.~network reliability estimation~\cite{duenas17count}.
	For SDD $\alpha$, model count can be performed with  $O(|\alpha|)$ time dynamic programming.
	Similarly, we can show that the model count of the function represented by a VS-SDD $(\alpha, k)$ can be computed by dynamic programming that runs in $\order{|\alpha|}$ time.
    The pseudocode for model counting with VS-SDD is shown in Alg.~\ref{alg:count}.
	The key is that substitution-equivalent Boolean functions have the same model count.

\begin{algorithm}[tbp]
	\caption{$\mathtt{Count}(\alpha,k)$, which computes the model count of the Boolean function $\langle\alpha,k\rangle$.}
	\label{alg:count}
	\fontsize{9pt}{7pt}\selectfont
	\begin{algorithmic}[1]
		\REQUIRE A decomposition VS-SDD $(\alpha,k)$.
		\ENSURE The model count of the Boolean function $\langle\alpha,k\rangle$.
		\IF{$\mathtt{Cache}(\alpha)\neq\mathtt{nil}$}
			\RETURN $\mathtt{Cache}(\alpha)$
		\ELSE
			\STATE $w\leftarrow\idtov(k)$
			\STATE $r\leftarrow 0$
			\FORALL{elements $([p_i,d_i],[s_i,e_i])$ in $\alpha$}
				\STATE \textbf{if} $p_i=\bot$ \textbf{then} $n_p\leftarrow 0$
				\STATE \textbf{else if} $p_i=\top$ \textbf{then} $n_p\leftarrow 2^{L(w^l)}$ \COMMENT{$w^l$ is the left child of $w$}
				\STATE \textbf{else if} $p_i\in\{\apos,\aneg\}$ \textbf{then} $n_p\leftarrow 2^{L(w^l)-1}$
				\STATE \textbf{else} $n_p\leftarrow\mathtt{Count}(p_i,k+d_i)\cdot 2^{L(w^l)-L(\idtov(k+d_i))}$
				\STATE \textbf{if} $s_i=\bot$ \textbf{then} $n_s\leftarrow 0$
				\STATE \textbf{else if} $s_i=\top$ \textbf{then} $n_s\leftarrow 2^{L(w^r)}$ \COMMENT{$w^r$ is the right child of $w$}
				\STATE \textbf{else if} $s_i\in\{\apos,\aneg\}$ \textbf{then} $n_s\leftarrow 2^{L(w^r)-1}$
				\STATE \textbf{else} $n_s\leftarrow\mathtt{Count}(s_i,k+e_i)\cdot 2^{L(w^r)-L(\idtov(k+e_i))}$
				\STATE $r\leftarrow r+n_pn_s$
			\ENDFOR
			\RETURN $\mathtt{Cache}(\alpha)\leftarrow r$
		\ENDIF
	\end{algorithmic}
	\normalsize
\end{algorithm}

	\textbf{ME} can also be solved by an algorithm similar to \textbf{CT}.
	For \textbf{SE}, we are given two VS-SDDs $(\alpha, k)$ and $(\beta, \ell)$, and check whether $\langle\alpha,k \rangle$ implies $\langle\beta,\ell\rangle$ or not.
	This can be solved by the algorithm shown in \cite{pipat08struct}.
	That is, we take the conjunction $\langle\alpha, k\rangle\wedge\langle\beta, \ell\rangle$, and then perform model counting.
	If the model count of this conjunction equals that of $\langle\alpha,k\rangle$, we can say $\langle\alpha,k\rangle$ implies $\langle\beta,\ell\rangle$.
	Since conjunction and model counting can be performed in polytime for VS-SDDs, \textbf{SE} can also be solved in polytime.
	Note that even for compressed VS-SDDs that do not support $\mathbf{\wedge}$\textbf{BC}, the procedure described above can be performed in polytime because during this procedure the conjunction VS-SDD is not obliged to be compressed, since it is only used for model counting.
	\textbf{EQ}, \textbf{CO}, \textbf{VA}, \textbf{IM}, and \textbf{CE} can be solved by \textbf{SE}.

	We next consider the transformations in Table \ref{tb:query}(b).
	First, the negation $\mathbf{\neg}$\textbf{C} of a VS-SDD $(\alpha, k)$ can be computed by taking exclusive-or with $\top$, which can be done in $\order{|\alpha|\cdot 1}=\order{|\alpha|}$ time and thus VS-SDDs support polytime $\mathbf{\neg}$\textbf{C}.
	Note that this procedure produces a compressed VS-SDD if $\alpha$ is also compressed since if $s_i\neq s_j$ for all $i\neq j$, $\neg s_i\neq\neg s_j$ for all $i\neq j$.
	Therefore compressed VS-SDDs also support polytime $\mathbf{\neg}$\textbf{C}.
	
	The negative results for VS-SDDs and compressed VS-SDDs in Table \ref{tb:query}(b) can be proved in a similar manner as those of SDDs and compressed SDDs in \cite{broeck15comp}.
	Thereafter, we focus on uncompressed VS-SDDs.

	Positive results for $\mathbf{\wedge}$\textbf{BC} and $\mathbf{\vee}$\textbf{BC} are exactly as stated in Prop.~\ref{prop:apply}.
	For \textbf{CD}, given VS-SDD $(\alpha, k)$ and term $S$ (a conjunction of literals), we return a VS-SDD representing $\langle\alpha, k\rangle|_S$, where $f|_S$ is the Boolean function obtained by replacing each occurrence of $X_i$ in $f$ with $\mathtrue$ if $S$ contains $X_i$, or with $\mathfalse$ if $S$ contains $\neg X_i$.
	We follow the procedure for conditioning in an uncompressed SDD as detailed in the full version of \cite{broeck15comp}.

	Conditioning may unpack a VS-SDD node to at most $(|S|+1)$ nodes if we apply the identical vtree rule, and we can perform conditioning in $\order{|S||\alpha|}$ time, which is still polynomial with regard to $|S|$ and $|\alpha|$.
	Thus VS-SDDs support polytime \textbf{CD}.
    The pseudocode for conditioning on VS-SDD is written in Alg.~\ref{alg:cond}.
	The support for \textbf{SFO} follows from the support for \textbf{CD} and $\mathbf{\vee}$\textbf{BC}.
\begin{algorithm}[tbp]
	\caption{$\mathtt{Cond}(\alpha,k,S)$, which performs a conditioning of $(\alpha,k)$ with the literals in $S$.}
	\label{alg:cond}
	\fontsize{9pt}{7pt}\selectfont
	\begin{algorithmic}[1]
		\REQUIRE A VS-SDD $(\alpha,k)$, and a set of literals $S$.
		\ENSURE A VS-SDD structure of the Boolean function obtained by conditioning $\langle\alpha,k_\alpha\rangle$ with the literals in $S$.
		\IF{$v(\idtov(k))$ does not contain variables appearing in $S$}
			\RETURN $\alpha$
		\ELSIF{$\mathtt{Cache}(\alpha,k)\neq\mathtt{nil}$}
			\RETURN $\mathtt{Cache}(\alpha,k)$
		\ELSE
			\STATE \textbf{if} $\alpha\in\{\top,\bot\}$ \textbf{then} \textbf{return} $\alpha$
			\STATE \textbf{else if} ($\alpha=\apos$ \textbf{and} $l(\idtov(k))\in S$) \textbf{or} ($\alpha=\aneg$ \textbf{and} $\neg l(\idtov(k))\in S$) \textbf{then} \textbf{return} $\top$
			\STATE \textbf{else if} ($\alpha=\apos$ \textbf{and} $\neg l(\idtov(k))\in S$) \textbf{or} ($\alpha=\aneg$ \textbf{and} $l(\idtov(k))\in S$) \textbf{then} \textbf{return} $\bot$
			\STATE $\gamma\leftarrow\{\}$
			\FORALL{elements $([p_i,d_{p_i}],[s_i,d_{s_i}])$ in $\alpha$}
				\STATE add element $([\mathtt{Cond}(p_i,k+d_{p_i},S),d_{p_i}],[\mathtt{Cond}(s_i,k+d_{s_i},S),d_{s_i}])$ to $\gamma$
			\ENDFOR
			\IF{$\mathtt{UniqTable}(e(\idtov(k)), \gamma)=\mathtt{nil}$}
				\STATE $\mathtt{UniqTable}(e(\idtov(k))), \gamma)\leftarrow \mathtt{CreateNewNode}(\gamma)$
			\ENDIF
			\RETURN $\mathtt{Cache}(\alpha,k)\leftarrow\mathtt{UniqTable}(e(\idtov(k)), \gamma)$
		\ENDIF
	\end{algorithmic}
	\normalsize
\end{algorithm}
\end{proof}

\section{Appendix: The Apply Operation for Trimmed VS-SDDs}
\label{appx:trimmedapply}
In this appendix, we detail the $\mathapply$ operation for trimmed VS-SDDs.
The simplicity of the $\mathapply$ for normalized VS-SDDs is due to the fact that we can assume that the offsets of two VS-SDDs are always equal.
For trimmed VS-SDDs, this assumption does not hold and thus we should consider the case that the offsets of two VS-SDDs are different.
Now the $\mathapply$ operation takes five arguments $\mathapply(\alpha,\beta,k_\alpha,k_\beta,\circ)$ to compute the VS-SDD of $\langle\alpha,k_\alpha\rangle\circ\langle\beta,k_\beta\rangle$.

To deal with the case $k_\alpha\neq k_\beta$, we should consider the \emph{expansion} at vtree node $w$.
Let $\mathbf{Z}$ and $\mathbf{W}$ be the variables corresponding to the left and right (resp.) descendant leaves of $w$.
Then, we can make the $\mathbf{Z}$-partition of the function $H_l(\mathbf{Z},\mathbf{W}):=h_l(\mathbf{Z})$ as $[h_l(\mathbf{Z})\wedge\mathtrue]\vee[(\neg h_l(\mathbf{Z}))\wedge\mathfalse]$ (called \emph{left expansion}).
We can also make the $\mathbf{Z}$-partition of the function $H_r(\mathbf{Z},\mathbf{W}):=h_r(\mathbf{W})$ as $[\mathtrue\wedge h_r(\mathbf{W})]$ (called \emph{right expansion}).
By following this, we can form a decomposition of a VS-SDD $(\alpha,k_\alpha)$ at the ancestor vtree node of $\idtov(k_\alpha)$ ($\beta$ can also be handled in the same way).
If $k_\alpha\neq k_\beta$, we expand either or both of $(\alpha,k_\alpha)$ and $(\beta,k_\beta)$ at the \emph{lowest common ancestor} (LCA) node of $\idtov(k_\alpha)$ and $\idtov(k_\beta)$ to make $\mathbf{X}$-partition with the same $\mathbf{X}$.
Note that for SDDs, $\mathapply$ is also based on the following mechanism.

The $\mathapply(\alpha,\beta,k_\alpha,k_\beta,\circ)$ procedure can be classified into five cases depending on the relation of $k_\alpha$ and $k_\beta$.
The full pseudocode for the VS-SDDs' $\mathapply$ is given in Alg.~\ref{alg:apply}.
Here for the simplicity, the offsets of constants are considered as $0$, and $\idtov(0)$ is considered as a right descendant of any other vtree node.
Note that the cases that $(\alpha,k_\alpha)$ and $(\beta,k_\beta)$ are exchanged can be handled in the same manner.
Let $w$ be the LCA of $\idtov(k_\alpha)$ and $\idtov(k_\beta)$.
Then the returned offset is $\vtoid(w)$, unless otherwise specified.
\begin{enumerate}
	\renewcommand{\theenumi}{(\arabic{enumi})}
	\renewcommand{\labelenumi}{(\arabic{enumi})}
	\item If both $(\alpha,k_\alpha)$ and $(\beta,k_\beta)$ are constants, either one is a constant and the other is a literal, or both are literals with $k_\alpha=k_\beta$, then the returned VS-SDD structure is $\alpha\circ\beta$, which is either constant or literal. For example, $\top\wedge\apos=\apos$ and $\apos\oplus\apos=\bot$. The returned offset is $0$ if $\alpha\circ\beta$ is a constant, and $\max\{k_\alpha,k_\beta\}$ otherwise.
	\item If $k_\alpha=k_\beta$, then $\vtoid(w)=k_\alpha=k_\beta$. Let $(\lambda_{i,j},k_{\lambda_{i,j}})$ be the returned VS-SDD of $\mathapply(p_i,q_j,k_\alpha+d_{p_i},k_\beta+d_{q_j},\wedge)$ and $(\mu_{i,j},k_{\mu_{i,j}})$ be that of $\mathapply(s_i,r_j,k_\alpha+d_{s_i},k_\beta+d_{r_j},\circ)$, where $([p_i,d_{p_i}],[s_i,d_{s_i}])\in\alpha$ and $([q_j,d_{q_j}],[r_j,d_{r_j}])\in\beta$.
	The resultant VS-SDD structure is $\{([\lambda_{i,j},\max\{k_{\lambda_{i,j}}-\vtoid(w),0\}],[\mu_{i,j},\max\{k_{\mu_{i,j}}-\vtoid(w),0\}])\mid i,j\}$.
	Note that $\max\{\cdot,0\}$ deals with the case that $\lambda_{i,j}$ or $\mu_{i,j}$ is a constant.
	\item If $\idtov(k_\alpha)$ is a left descendant of $\idtov(k_\beta)$, then $\vtoid(w)=k_\beta$.
	Here $\alpha$ is \emph{left expanded} to $\alpha'=\{([\alpha,0],[\top,0]),([\neg\alpha,0],[\bot,0])\}$, and the same computation as case (2) runs except that $\alpha$ is replaced with $\alpha'$.
	\item If $\idtov(k_\alpha)$ is a right descendant of $\idtov(k_\beta)$, then $\vtoid(w)=k_\beta$. 
	Here $\alpha$ is \emph{right expanded} to $\alpha'=\{([\top,0],[\alpha,0])\}$, and the same computation as case (2) runs except that $\alpha$ is replaced with $\alpha'$.
	\item If $\idtov(k_\alpha)$ and $\idtov(k_\beta)$ are left and right descendants of $w$ (resp.),
	$\alpha$ is left expanded to $\alpha'$, $\beta$ is right expanded to $\beta'$, and the same computation as case (2) runs except that $\alpha$ and $\beta$ are replaced with $\alpha'$ and $\beta'$ (resp.)
\end{enumerate}

Here we analyze the time complexity of the $\mathapply$ algorithm.
Now for each $(\gamma,\delta,k_\gamma,k_\delta,\circ)$, the cost of the $\mathapply$ call other than the recursion is bounded by $\order{t_\gamma t_\delta}$ where $t_\gamma$ and $t_\delta$ are the decomposition sizes of $\gamma$ and $\delta$ (resp.), and there are at most $\order{M}$ candidates for each of the offsets $k_\gamma,k_\delta$.
Here cases (3) and (5) must deal with the negation of a VS-SDD node, but this only increases the complexity by a constant factor; the details are described later.
However, it seems that the overall cost of $\mathapply(\alpha,\beta,k_\alpha,k_\beta,\circ)$ can only be bounded by $\order{\sum_{\gamma\in\alpha}\sum_{\delta\in\beta}t_\gamma t_\delta}\cdot\order{M^2}=\order{M^2|\alpha||\beta|}$, which is not a polytime of $|\alpha|$ and $|\beta|$.
Note that the $\mathapply$ described above needs LCA indexing~\cite{harel84LCA} of a vtree, which needs $\order{M}$ preprocessing time where $M$ is the number of variables\footnote{The $\mathapply$ of (trimmed) SDDs also needs such LCA indexing. Typically, the same vtree is repetitively used many times, and so such LCA indexing is considered to be just a preprocessing step for SDDs' $\mathapply$ in \cite{darwiche11sdd}. Therefore, we also consider this cost as a preprocessing step for VS-SDDs' $\mathapply$.}.

However, we can omit some ``isomorphic'' computations with VS-SDDs, as described in the main article before Prop.~\ref{prop:apply}.
More formally, the key is the following lemma.
From now, for two vtree nodes $u, w$ in the vtree $v$, $u\sim w$ means that the subtree rooted at $u$ and that rooted at $w$ are isomorphic.

\begin{lemma}
	\label{lem:applycong}
	Given vtree $v$, and two VS-SDD structures $\alpha,\beta$, we consider nodes $\gamma\in\alpha$ and $\delta\in\beta$.
	Let $k_\gamma, k'_\gamma$ be the possible offsets of $\gamma$ and $k_\delta,k'_\delta$ be those of $\delta$.
	Let $w$ be the LCA of $\idtov(k_\gamma)$ and $\idtov(k_\delta)$ and $w'$ be that of $\idtov(k'_\gamma)$ and $\idtov(k'_\delta)$.
	Then if (I) $\gamma\in\{\top,\bot\}$ or $k_\gamma-\vtoid(w)=k'_\gamma-\vtoid(w')$, (II) $\delta\in\{\top,\bot\}$ or $k_\delta-\vtoid(w)=k'_\delta-\vtoid(w')$, and (III) $w\sim w'$, $\mathapply(\gamma,\delta,k_\gamma,k_\delta,\ldots)$ and $\mathapply(\gamma,\delta,k'_\gamma,k'_\delta,\ldots)$ result in an identical structure.
\end{lemma}
\begin{proof}
	The proof is by the induction of the depths of $w$ and $w'$ in $v$ (note that since $w\sim w'$, $w$ and $w'$ have the same depth).
	The base case is that $w$ and $w'$ are both leaves or $\vtoid(w)=\vtoid(w')=0$, which corresponds to case (1) and thus holds trivially.
	
	The step case is that $w$ and $w'$ are internal nodes.
	First, we deal with the case $k_\gamma=k_\delta=\idtov(w)$, which corresponds to case (2).
	Then from conditions (I) and (II), $k'_\gamma=k'_\delta=\idtov(w')$, which also corresponds to case (2).
	Let $(\lambda_{i,j},k_{\lambda_{i,j}})=\mathapply(p_i,q_j,k_\gamma+d_{p_i},k_\delta+d_{q_j},\wedge)$, $(\mu_{i,j},k_{\mu_{i,j}})=\mathapply(s_i,r_j,k_\gamma+d_{s_i},k_\delta+d_{r_j},\circ)$ and let $(\lambda'_{i,j},k_{\lambda'_{i,j}})=\mathapply(p_i,q_j,k'_\gamma+d_{p_i},k'_\delta+d_{q_j},\wedge)$, $(\mu'_{i,j},k_{\mu'_{i,j}})=\mathapply(s_i,r_j,k'_\gamma+d_{s_i},k'_\delta+d_{r_j},\circ)$, where $([p_i,d_{p_i}],[s_i,d_{s_i}])\in\gamma$, $([q_j,d_{q_j}],[r_j,d_{r_j}])\in\delta$.
	Then $\mathapply(\gamma,\delta,k_\gamma,k_\delta,\circ)$ and $\mathapply(\gamma,\delta,k'_\gamma,k'_\delta,\circ)$ are
	\begin{align*}
		& (\{([\lambda_{i,j},\max\{k_{\lambda_{i,j}}-\vtoid(w),0\}],[\mu_{i,j},\max\{k_{\mu_{i,j}}-\vtoid(w),0\}])\mid i,j\},\vtoid(w))\ \text{and}\\
		& (\{([\lambda'_{i,j},\max\{k_{\lambda'_{i,j}}-\vtoid(w'),0\}],[\mu'_{i,j},\max\{k_{\mu'_{i,j}}-\vtoid(w'),0\}])\mid i,j\},\vtoid(w')),
	\end{align*}
	respectively.
	Since $w\sim w'$ (this suggests the topologies of the subtrees rooted at $w$ and $w'$ are identical), $\idtov(k_{\lambda_{i,j}})=\mathtt{LCA}(\idtov(k_\alpha+d_{p_i}),\idtov(k_\beta+d_{q_j}))$ and $\idtov(k_{\lambda'_{i,j}})=\mathtt{LCA}(\idtov(k'_\alpha+d_{p_i}),\idtov(k'_\beta+d_{q_j}))$ satisfies $\idtov(k_{\lambda_{i,j}})\sim \idtov(k_{\lambda'_{i,j}})$ and $k_{\lambda_{i,j}}-\vtoid(w)=k_{\lambda'_{i,j}}-\vtoid(w')$.
	Thus we can use the induction hypothesis for $(\lambda_{i,j},k_{\lambda_{i,j}})$ and $(\lambda'_{i,j},k_{\lambda'_{i,j}})$: $\lambda_{i,j}$ and $\lambda'_{i,j}$ result in an identical structure.
	Similarly, $\mu_{i,j}$ and $\mu'_{i,j}$ result in an identical structure and $k_{\mu_{i,j}}-\vtoid(w)=k_{\mu'_{i,j}}-\vtoid(w')$.
	Therefore $\mathapply(\gamma,\delta,k_\gamma,k_\delta,\circ)$ and $\mathapply(\gamma,\delta,k'_\gamma,k'_\delta,\circ)$ also result in an identical structure.
	
	The other cases ((3), (4) and (5)) can be treated in almost the same way.
	Note that case (4) involves the case in which $\gamma,\gamma'$ are constants but $\delta,\delta'$ are non-constants.
	In this case, $w=\idtov(k_\delta)$, $w'=\idtov(k'_\delta)$ and consequently the same argument holds.
\end{proof}

Now the pseudocode of the $\mathapply$ of VS-SDDs can be written as Alg.~\ref{alg:apply}.
Here the Boolean variables $f_\alpha$ and $f_\beta$ are additionally included in the arguments of $\mathapply$, because in the left expansion (appearing in cases (3) and (5)) the negation of a VS-SDD node should be considered.
Here we stress that this only increases the computational cost by a constant factor, and thus the asymptotical complexity does not change.
Note that such handlings of negation should also be needed for the SDDs' $\mathapply$.

Now we explain the pseudocode.
Here for simplicity, the offset of the constant VS-SDDs is considered to be always $0$.
Lines 1--2 deal with constants and literals with negation flag; for them, the negation can be easily handled, e.g.~$(\alpha,f_\alpha)=(\aneg,\mathtrue)$ is converted into $(\apos,\mathfalse)$.
Lines 3--4 specify the offset of the constants to $0$.
Line 5 computes the LCA, $w$, of $\idtov(k_\alpha)$ and $\idtov(k_\beta)$, which can be computed in $\order{1}$ time with $\order{M}$ preprocessing for the vtree~\cite{harel84LCA}, and Lines 6--7 compute the differences of vtree node ids corresponding to conditions (I) and (II) in Lemma \ref{lem:applycong}.
Note that in Line 6, if $(\alpha,k_\alpha)$ is constant, i.e.~$k_\alpha=0$, then $e_\alpha=0$, and otherwise $e_\alpha=k_\alpha-\vtoid(w)$.
Lines 8--11 correspond to case (1).
For example, $\top\wedge\apos=\apos$ and $\apos\oplus\apos=\bot$.
Lines 12--13 are important: since when fixing $(\alpha,\beta,f_\alpha,f_\beta)$, if $e_\alpha$, $e_\beta$ and $e(w)=\argmin\{\vtoid(u)\mid u\sim w\}$ are equal, then the resultant structures are identical due to Lemma \ref{lem:applycong}, the computation cache is called, and if already computed the computed result is returned with offset $\vtoid(w)$.
If not computed, $(\alpha,k_\alpha)$ and $(\beta,k_\beta)$ are left or right expanded if needed (see Alg.~\ref{alg:expand}).
That is, if $\idtov(k_\alpha)$ is a left (resp.~right) descendant of $w$ then $(\alpha,k_\alpha)$ is left (resp.~right) expanded in Line 14.
The expansion of $\beta$ (Line 15) is conducted in the same way.
After that, the node representing $\langle\alpha,k_\alpha\rangle\circ\langle\beta,k_\beta\rangle$ is recursively computed in Lines 17--22.
Here the formula $\max\{\cdot,0\}$ in Line 22 deals with the case the computed $p$ (or $s$) is a constant.
In Line 20, it is checked if the computed prime $p$ is $\mathfalse$.
If $p=\mathfalse$, the corresponding sub is not computed since such $(p,s)$ pair does not constitute an $\mathbf{X}$-partition.
Such checking can be performed via the $\mathtt{Consistent}$ algorithm described in Alg.~\ref{alg:consistent}.
It runs in time linear to the size of its decomposition (other than the recursion), and thus with the power of cache ($\mathtt{Cache}(\cdot)$ in Alg.~\ref{alg:consistent}), the total cost of calling $\mathtt{Consistent}$ is bounded linear to the resultant structure of $\mathapply$, which does not incur an increase on the time complexity of $\mathapply$.
The hash $\mathtt{UniqTable}$ returns the output node if an identical substructure satisfying the identical vtree rule is already constructed.
If not yet constructed, the decomposition node with $\gamma$ is generated in Line 24.

\begin{algorithm}[tbp]
	\caption{$\mathtt{Apply}(\alpha,\beta,k_\alpha,k_\beta,f_\alpha,f_\beta,\circ)$, which computes a VS-SDD representing $\langle\alpha,k_\alpha\rangle\circ\langle\beta,k_\beta\rangle$ for two VS-SDDs $(\alpha,k_\alpha), (\beta,k_\beta)$ and a binary operator $\circ$.}
	\label{alg:apply}
	\fontsize{9pt}{7pt}\selectfont
	\begin{algorithmic}[1]
		\REQUIRE VS-SDDs $(\alpha,k_\alpha), (\beta,k_\beta)$, Boolean values $f_\alpha, f_\beta$, a binary operator $\circ$.
		\ENSURE A VS-SDD representing $\langle\alpha,k_\alpha\rangle\circ\langle\beta,k_\beta\rangle$. If $f_\alpha=\mathtrue$ ($f_\beta=\mathtrue$), $\langle\alpha,k_\alpha\rangle$ ($\langle\beta,k_\beta\rangle$) is replaced by $\neg\langle\alpha,k_\alpha\rangle$ ($\neg\langle\beta,k_\beta\rangle$). If the output VS-SDD structure is either $\top$ or $\bot$, the output offset is 0.
		\STATE \textbf{if} $\alpha\in\{\top,\bot,\mathbf{v},\neg\mathbf{v}\}$ \AND $f_\alpha=\mathtrue$ \textbf{then} $\alpha\leftarrow\neg\alpha$, $f_\alpha\leftarrow\mathfalse$
		\STATE \textbf{if} $\beta\in\{\top,\bot,\mathbf{v},\neg\mathbf{v}\}$ \AND $f_\beta=\mathtrue$ \textbf{then} $\beta\leftarrow\neg\beta$, $f_\beta\leftarrow\mathfalse$
		\STATE \textbf{if} $\alpha\in\{\top,\bot\}$ \textbf{then} $k_\alpha\leftarrow 0$
		\STATE \textbf{if} $\beta\in\{\top,\bot\}$ \textbf{then} $k_\beta\leftarrow 0$
		\STATE $w\leftarrow\mathtt{LCA}(\idtov(k_\alpha),\idtov(k_\beta))$
		\STATE $e_\alpha\leftarrow\max\{k_\alpha-\vtoid(w),0\}$
		\STATE $e_\beta\leftarrow\max\{k_\beta-\vtoid(w),0\}$
		\IF{$\alpha,\beta\in\{\top,\bot,\mathbf{v},\neg\mathbf{v}\}$ \AND ($k_\alpha=0$ \OR $k_\beta=0$ \OR $k_\alpha=k_\beta$)}
			\STATE $\gamma\leftarrow\alpha\circ\beta$ \COMMENT{This computation must result in $\top,\bot,\mathbf{v}$ or $\neg\mathbf{v}$}
			\STATE \textbf{if} $\gamma\in\{\top,\bot\}$ \textbf{then} \textbf{return} $(\gamma,0)$
			\STATE \textbf{else} \textbf{return} $(\gamma,\max\{k_\alpha,k_\beta\})$
		\ELSIF{$\mathtt{ConvTable}(\alpha,\beta,f_\alpha,f_\beta,e_\alpha,e_\beta,e(w),\circ)\neq\mathtt{nil}$}
			\RETURN $(\mathtt{ConvTable}(\alpha,\beta,f_\alpha,f_\beta,e_\alpha,e_\beta,e(w),\circ),\vtoid(w))$
		\ENDIF
		\STATE $\alpha'\leftarrow\mathtt{Expand}(\alpha,f_\alpha,k_\alpha,w)$
		\STATE $\beta'\leftarrow\mathtt{Expand}(\beta,f_\beta,k_\beta,w)$
		\STATE $\gamma\leftarrow\{\}$
		\FORALL{elements $([p_i,d_{p_i}],[s_i,d_{s_i}],f_{p_i},f_{s_i})$ in $\alpha'$}
			\FORALL{elements $([q_j,d_{q_j}],[r_j,d_{r_j}],f_{q_j},f_{r_j})$ in $\beta'$}
				\STATE $(p,k_p)\leftarrow\mathtt{Apply}(p_i,q_j,k_\alpha+d_{p_i},k_\beta+d_{q_j},f_{p_i},f_{q_j},\wedge)$
				\IF{$\mathtt{Consistent}(p)$}
					\STATE $(s,k_s)\leftarrow\mathtt{Apply}(s_i,r_j,k_\alpha+d_{s_i},k_\beta+d_{r_j},f_{s_i},f_{r_j},\circ)$
					\STATE add element $([p,\max\{k_p-\vtoid(w),0\}],[s,\max\{k_s-\vtoid(w),0\}])$ to $\gamma$
				\ENDIF
			\ENDFOR
		\ENDFOR
		\IF{$\mathtt{UniqTable}(e(w), \gamma)=\mathtt{nil}$}
			\STATE $\mathtt{UniqTable}(e(w), \gamma)\leftarrow \mathtt{CreateNewNode}(\gamma)$
		\ENDIF
		\RETURN $\mathtt{ConvTable}(\alpha,\beta,f_\alpha,f_\beta,e_\alpha,e_\beta,e(w),\circ)\leftarrow\mathtt{UniqTable}(e(\idtov(j)), \gamma)$
	\end{algorithmic}
	\normalsize
\end{algorithm}

\begin{algorithm}[tbp]
	\caption{$\mathtt{Expand}(\alpha,f,k,w)$, which expands VS-SDD $(\alpha,k)$ at the level of vtree node $w$ and returns the structure.}
	\label{alg:expand}
	\fontsize{9pt}{7pt}\selectfont
	\begin{algorithmic}[1]
	\IF{$k=\vtoid(w)$}
		\RETURN $\{([p_i,d_{p_i}],[s_i,d_{s_i}],\mathfalse,f_\alpha)\mid([p_i,d_{p_i}],[s_i,d_{s_i}])\in\alpha\}$
	\ELSIF[$w^r$ is the right child of $w$]{$k<\vtoid(w^r)$ \AND $k\neq 0$}
		\RETURN $\{([\alpha,0],[\top,0],f,\mathfalse),([\alpha,0],[\bot,0],\neg f,\mathfalse)\}$ \COMMENT{left expansion}
	\ELSE
		\RETURN $\{([\top,0],[\alpha,0],\mathfalse,f)\}$ \COMMENT{right expansion}
	\ENDIF
	\end{algorithmic}
	\normalsize
\end{algorithm}

\begin{algorithm}[tbp]
	\caption{$\mathtt{Consistent}(\alpha)$, which decides whether the Boolean function that the VS-SDD $\alpha$ represents is not $\mathfalse$.}
	\label{alg:consistent}
	\fontsize{9pt}{7pt}\selectfont
	\begin{algorithmic}[1]
		\REQUIRE VS-SDD $\alpha$.
		\ENSURE $\mathfalse$ if the Boolean function $\alpha$ represents is $\mathfalse$, or $\mathtrue$ otherwise.
		\IF{$\alpha=\bot$}
			\RETURN $\mathfalse$
		\ELSIF{$\alpha\in\{\top,\apos,\aneg\}$}
			\RETURN $\mathtrue$
		\ELSIF{$\mathtt{Cache}(\alpha)\neq\mathtt{nil}$}
			\RETURN $\mathtt{Cache}(\alpha)$
		\ELSE
			\FORALL{elements $([p_i,d_{p_i}],[s_i,d_{s_i}])$ in $\alpha$}
				\IF{$\mathtt{Consistent}(p_i)$ \AND $\mathtt{Consistent}(q_i)$}
					\RETURN $\mathtt{Cache}(\alpha)\leftarrow\mathtrue$
				\ENDIF
			\ENDFOR
			\RETURN $\mathtt{Cache}(\alpha)\leftarrow\mathfalse$
		\ENDIF
	\end{algorithmic}
	\normalsize
\end{algorithm}

Using Lemma \ref{lem:applycong}, the time complexity can be proved as follows.

\begin{proof}[Proof of Prop.~\ref{prop:apply} for trimmed VS-SDDs]
    Since the cost of computations involving $\top$ and $\bot$ is absorbed in other costs, we consider that among literals and decomposition nodes.
	For literals ($\apos$ and $\aneg$) or decomposition nodes $\lambda$, let $t_\lambda$ be the size of decomposition (here $t_{\mathbf{v}}=t_{\neg\mathbf{v}}=0$) and $T_\lambda$ be the number of incoming edges of $\lambda$ (here for the root VS-SDD node $r$, let $T_r=1$).
	Then we observe that $\sum_{\gamma\in\alpha}t_{\gamma}=|\alpha|$ and $\sum_{\gamma\in\alpha}T_{\gamma}=\order{|\alpha|}$, and as is the case with $\beta$.
	Now we analyze the total cost of all calls $\mathapply(\gamma,\delta,k_\gamma,k_\delta,\ldots)$ (other than the recursion) for fixed $\gamma\in\alpha$ and $\delta\in\beta$ (but varying $k_\gamma$ and $k_\delta$) when calling $\mathapply(\alpha,\beta,k_\alpha,k_\beta,\mathfalse,\mathfalse,\circ)$ at the top level.
	Let $w_\gamma=\idtov(k_\gamma)$ and $w_\delta=\idtov(k_\delta)$.
	Then there are multiple possibilities for $w_\gamma$ and $w_\delta$ as described above.
	However, from the identical vtree rule, the candidates of $w_\gamma$ are all equivalent up to the relation $\sim$, and as is the case with $w_\delta$.
	Now we divide the pair of the candidates of $(w_\gamma,w_\delta)$ into four cases.
	
	(i) $w_\gamma=w_\delta$.
	This corresponds to case (2).
	For this case, it takes $\order{t_{\gamma}t_{\delta}}$ time to compute if Line 13 is not executed.
	Note that if Line 13 is executed, the cost of this $\mathapply$ call is absorbed in that of the preceding $\mathapply$ call.
	Since in this case $w_{\gamma}=w_{\delta}=w$ and the candidates of $w_\gamma$ are equivalent up to $\sim$, the conditions (I)--(III) of Lemma \ref{lem:applycong} are satisfied among the candidates of $w_\gamma=w_\delta$, and thus in this case we need to proceed after Line 14 only once.
	Therefore, the total cost of this type of computation for fixed $\gamma$ and $\delta$ is bounded by $\order{t_\gamma t_\delta}$.
	
	(ii) $w_\gamma$ is a (proper) descendant of $w_\delta$.
	This corresponds to cases (3) and (4).
	For this case, it takes $\order{t_\delta}$ time per one call if Line 13 is not executed, since $\gamma$ is expanded to a decomposition of constant size.
	If $\gamma$ is the root node of $\alpha$, it is trivial that $\mathapply(\gamma,\delta,\ldots)$ is called only a constant number of times, since $k_\gamma=k_\alpha$, and among many candidates of $k_\delta$, the condition $w_\delta$ is an ancestor of $w_\gamma$ uniquely determines $k_\delta$.
	Otherwise, let $\lambda$ be one of the parent nodes of $\gamma$ (i.e.~nodes such that the decomposition has $[\gamma,d_{*}]$ as a prime or a sub).
	Now consider the case $\mathapply(\lambda,\cdot,k_\lambda,\cdot,\ldots)$ precedes $\mathapply(\gamma,\delta,k_\gamma,k_\delta,\ldots)$, i.e.~the situation the edge directed from $\lambda$ to $\gamma$ is traversed.
	Then we claim the following.
	\begin{claim}
		\label{clm:parent}
		$w_\lambda:=\idtov(k_\lambda)$ is a proper ancestor of $w_\delta$.
	\end{claim}
	\begin{claimproof}
		Suppose $w_\lambda$ is a (not necessarily proper) descendant of $w_\delta$.
		Then, $\mathapply(\lambda,\delta,k_\lambda,k_\delta,\ldots)$ should be called before $\mathapply(\gamma,\ldots)$, and since the decomposition of $\delta$ is processed in this call (i.e.~$\delta$ is not expanded), $\mathapply(\gamma,\delta,k_\gamma,k_\delta,\ldots)$ is not called.
		Suppose $w_\lambda$ is neither an ancestor nor a descendant of $w_\delta$.
		Then, since $\idtov(k_\gamma)$ is a descendant of $w_\lambda$, it is also neither an ancestor nor a descendant of $w_\delta$, which contradicts the assumption.
	\end{claimproof}
	
	For all candidates of $w_\lambda$, the relative position of $w_\gamma$ compared to $w_\lambda$ (i.e.~$k_\gamma-k_\lambda$) is always equal to the difference $d_{p_i}$ or $d_{s_i}$.
	Moreover, since all candidates of $w_\lambda$ are equivalent up to $\sim$, and $w_\delta$ is a descendant of $w_\lambda$ and an ancestor of $w_\gamma$, the relative position of $w_\delta$ compared to $w_\lambda$ (i.e.~$k_\delta-k_\lambda$) is also always equal.
	Thus $k_\gamma-k_\delta=k_\gamma-\vtoid(w)$ is always equal, which satisfies condition (I) of Lemma \ref{lem:applycong}.
	Therefore, we need to proceed after Line 14 only once given that the call $\mathapply(\lambda,\cdot,k_\lambda,\cdot,\ldots)$ precedes.
	Since $\gamma$ has $T_\gamma$ parents, the total cost of this type of computation for fixed $\gamma$ and $\delta$ is bounded by $\order{T_\gamma t_\delta}$.
	
	(iii) $w_\delta$ is a descendant of $w_\gamma$.
	By reversing the argument of (ii), the total cost of this type of computation for fixed $\gamma$ and $\delta$ is bounded by $\order{t_\gamma T_\delta}$.
	
	(iv) $w_\gamma$ and $w_\delta$ have no ancestor-descendant relation.
	This corresponds to case (5).
	Even if Line 13 is not executed, it takes only constant time (other than the recursion) since both $\gamma$ and $\delta$ are expanded to constant size decompositions.
	If $\gamma$ and $\delta$ are root nodes of $\alpha$ and $\beta$, respectively, $\mathapply(\gamma,\delta,\ldots)$ is called only once, since $k_\gamma=k_\alpha$ and $k_\delta=k_\beta$.
	Otherwise, there must be a preceding $\mathapply$ call.
	The preceding $\mathapply$ call does not fall into case (iv), because once case (5) occurs, successive computations must involve constants.
	Thus, the cost of this $\mathapply$ call is absorbed in that of the preceding $\mathapply$ call.
	
	Now the total cost of $\mathapply(\alpha,\beta,k_\alpha,k_\beta,\mathfalse,\mathfalse,\circ)$ is $\sum_{\gamma\in\alpha}\sum_{\delta\in\beta}\order{t_\gamma t_\delta+T_\gamma t_\delta+t_\gamma T_\delta}=\order{|\alpha||\beta|}$, which proves Prop.~\ref{prop:apply}.
\end{proof}
\end{document}